\documentclass[conference]{IEEEtran}
\IEEEoverridecommandlockouts

\usepackage{cite}
\usepackage{amsmath,amssymb,amsfonts,amsthm}
\usepackage{algorithmic}
\usepackage{graphicx}
\usepackage{textcomp}
\usepackage{xcolor}
\usepackage[shortlabels]{enumitem}
\usepackage{hyperref}

\newcommand{\mat}[1]{\mathbf{#1}}
\newcommand{\arrayprod}[2]{\mathbin{{#1}.{#2}}}
\newcommand{\real}{\mathop{\mathrm{Re}}}
\newcommand{\imaginary}{\mathop{\mathrm{Im}}}
\newcommand{\nnz}{\mathop{\mathrm{nnz}}}
\newcommand{\powerset}{\mathop{\mathrm{Pow}}}
\newcommand{\coloneq}{\mathrel{{{:}{=}}}}
\newcommand{\imaginaryi}{\mathrm{i}\mkern1mu}
\newcommand{\concat}{{}^\frown}

\newtheorem{thm}{Theorem}
\newtheorem{prop}{Proposition}[thm]
\newtheorem{lem}{Lemma}[thm]

\allowdisplaybreaks
    
\begin{document}

\title{GraphBLAS Mathematical Opportunities: Parallel Hypersparse, Matrix Based Graph Streaming, and Complex-Index Matrices
\thanks{
Research was sponsored by the Department of the Air Force Artificial Intelligence Accelerator and was accomplished under Cooperative Agreement Number FA8750-19-2-1000. The views and conclusions contained in this document are those of the authors and should not be interpreted as representing the official policies, either expressed or implied, of the Department of the Air Force or the U.S. Government. The U.S. Government is authorized to reproduce and distribute reprints for Government purposes notwithstanding any copyright notation herein.
}
}

\author{\IEEEauthorblockN{
Hayden Jananthan, Jeremy Kepner, Michael Jones, Vijay Gadepally,  \\ Michael Houle, Peter Michaleas, Chasen Milner, Alex Pentland
\\
\IEEEauthorblockA{
MIT
}}}
\maketitle

\begin{abstract}
    The GraphBLAS  high performance library standard has yielded capabilities beyond enabling graph algorithms to be readily expressed in the language of linear algebra. These GraphBLAS capabilities enable new performant ways of thinking about algorithms that include leveraging hypersparse matrices for parallel computation, matrix-based graph streaming, and complex-index matrices.  Formalizing these concepts mathematically provides additional opportunities to apply GraphBLAS to new areas. This paper formally develops parallel hypersparse matrices, matrix-based graph streaming, and complex-index matrices and illustrates these concepts with various examples to demonstrate their potential merits.  
\end{abstract}


\section{Introduction}

The GraphBLAS (Graph Basic Linear Algebra Subprograms) are a good illustration of the  synergistic interplay between mathematics and computational capabilities.  New mathematical ideas drive the development of new capabilities in mathematical software. Likewise, the broad availability of high performance mathematical software with new capabilities can inspire the broader use of their underlying mathematical ideas.   The GraphBLAS standard effort \cite{Mattson2013} developed out of a growing awareness of the utility of graph algorithms in the language of linear algebra \cite{kepner2011graph} leading to further refinement of the mathematical foundations \cite{kepner16mathematical}, programming specifications \cite{bulucc2017design, brock2021introduction, kimmerer2024graphblas}, and implementations \cite{bulucc2011combinatorial, kepner2012dynamic, sundaram2015graphmat, hutchison2015graphulo, zhang2016gbtl, kumar2018ibm, davis2019algorithm, mcmillan2018design}.    

The GraphBLAS has yielded capabilities beyond enabling graph algorithms to be readily expressed in the language of linear algebra.  The mathematical foundations \cite{kepner2011graph}, derived from a more general associative array algebra \cite{kepner2018mathematics}, enabled the GraphBLAS implementations to incorporate hypersparse matrices where the number of nonzero ($\nnz$) entries is much less than the largest dimension \cite{buluc2009parallel}.  Hypersparse matrices in turn have enabled novel approaches to parallel computing that have a played a key role in the analysis of some of the largest corpora of network data \cite{
trigg2022hypersparse, jones2023deployment, jananthan2024anonymized, kepner2024what}.  Similarly, as matrices are the foundation for digital signal processing \cite{oppenheim1975digital, oppenheim1978applications, hayes1996statistical, smith1997scientist, manolakis2011applied}, the GraphBLAS have leveraged these same concepts streaming graph analysis \cite{gadepally2018hyperscaling, kepner19streaming
}.    More recently many interesting uses of complex networks with complex weights have emerged \cite{bottcher2024complex} that offer potential new applications for the GraphBLAS.

This work is part of an ongoing sequence exploring pure-math concepts with potential relevance to big data computations, such as, matrix-graph duality \cite{kepner2011graph, kepner16mathematical, fu2023algebraic}, relational algebra \cite{jananthan2017polystore}, associative arrays \cite{kepner2018mathematics}, incidence matrices \cite{jananthan2017constructing}, semirings \cite{kepner2021mathematics, lee2024eigenvalue}, and fuzzy algebra \cite{min2023fuzzy}.

The structure of the paper is as follows. \S \ref{mathematical preliminaries} covers the necessary mathematical concepts and notation needed for the paper; \S \ref{hypersparse addition} introduces and examines the mathematical background of using the summation of a sequence of hypersparse matrices as a partitioning mechanism for use in parallel, distributed algorithms, along with explicit examples given in the field of cyber networks; \S \ref{matrix-based streaming} introduces and examines the mathematical background of maxtrix-based graph streaming; finally, \S \ref{complex-index matrices} examines the mathematical background of using novel complex number-like values as a mechanism for incorporating the indices of a matrix into the entries to expand the types of algorithms expressible in the language of linear algebra, alongside several explicit examples. Formal proofs that the proposed semirings satisfy the requisite semiring axioms are left to \hyperref[appendix A]{Appendix A}, while proofs of their claimed additional properties are left to Appendix~\hyperref[appendix B]{Appendix B}.

\section{Mathematical Preliminaries}
\label{mathematical preliminaries}

In this section will define some necessary mathematical concepts and notation that underlie the rest of the paper. A central GraphBLAS mathematical concept is that of a \emph{semiring} \cite{golan2013semirings}. Formally, a semiring is an ordered quintuple $(\mathbb{S}, \oplus, \otimes, 0^\mathbb{S}, 1^\mathbb{S})$ consisting of an underlying set $\mathbb{S}$, two binary operations $\oplus, \otimes \colon \mathbb{S} \times \mathbb{S} \to \mathbb{S}$, and two constants $0^\mathbb{S}, 1^\mathbb{S}$ satisfying:
\begin{enumerate}[(i)]
	\item $\oplus$ is commutative, associative, and has identity $0^\mathbb{S}$.
	\item $\otimes$ is associative and has identity $1^\mathbb{S}$.
	\item $\otimes$ distributes over $\oplus$.
	\item $0^\mathbb{S}$ is an annihilator for $\otimes$.
\end{enumerate}
For ease of presentation, the underlying set $\mathbb{S}$ will often be used to stand for the semiring $(\mathbb{S}, \oplus, \otimes, 0^\mathbb{S}, 1^\mathbb{S})$ as a whole. Alternative definitions of semirings exist which either drop the necessity of $0^\mathbb{S}$ and/or $1^\mathbb{S}$ or weaken the necessary properties of $\otimes$. In particular, the GraphBLAS high performance library standard drops the requirement that $\otimes$ be associative; moreover, it introduces the notion of a `GraphBLAS semiring' which goes further, allowing $\otimes$ to operate on a potentially different domain but produce outputs in the domain of $\oplus$, dropping the identities $0^\mathbb{S}$ and $1^\mathbb{S}$ altogether, and dropping the requirement that $\otimes$ distribute over $\oplus$ \cite{brock2023graphblas}.

By an `array' we mean an \emph{associative array} \cite{kepner2018mathematics}, a map $\mat{A} \colon K_1 \times K_2 \to \mathbb{S}$ for which $\mat{A}(k_1, k_2) \neq 0^\mathbb{S}$ for at most finitely many pairs $(k_1, k_2) \in K_1 \times K_2$.\footnote{If $K_1$ and $K_2$ are finite sets, as is the case in any practical implementation or scenario, then this requirement is automatically satisfied.} Given an array $\mat{A} \colon K_1 \times K_2 \to \mathbb{S}$, $\nnz(\mat{A})$ denotes the number of pairs $(k_1, k_2) \in K_1 \times K_2$ for which $\mat{A}(k_1, k_2) \neq 0^\mathbb{S}$. 

Given arrays $\mat{A} \colon K_1 \times K_2 \to \mathbb{S}$ and $\mat{B} \colon K_3 \times K_4 \to \mathbb{S}$, we define arrays $\mat{A} \oplus \mat{B} \colon (K_1 \cup K_3) \times (K_2 \cup K_4) \to \mathbb{S}$ (element-wise addition), $\mat{A} \otimes \mat{B} \colon (K_1 \cap K_3) \times (K_2 \cap K_4) \to \mathbb{S}$ (element-wise multiplication), and $\mat{A} \arrayprod{\oplus}{\otimes} \mat{B} \colon K_1 \times K_4 \to \mathbb{S}$ (array multiplication) as follows. Given any $k_1, k_2$ chosen from appropriate domains\footnotemark:
\begin{align*}
	(\mat{A} \oplus \mat{B})(k_1, k_2) & \coloneq \mat{A}(k_1, k_2) \oplus \mat{B}(k_1, k_2), \\ 
	\\
	(\mat{A} \otimes \mat{B})(k_1, k_2) & \coloneq \mat{A}(k_1, k_2) \otimes \mat{B}(k_1, k_2), \\ 
	(\mat{A} \arrayprod{\oplus}{\otimes} \mat{B})(k_1, k_2) & \coloneq \bigoplus_{k_3 \in K_2 \cap K_3}{(\mat{A}(k_1, k_3) \otimes \mat{B}(k_3, k_2))}. 
\end{align*}
We remark that the well-definedness of array multiplication makes use of many of the properties of semirings: the associativity of $\oplus$ to make sense of the iterated binary operator $\bigoplus$, the fact that $0^\mathbb{S}$ is an annihilator for $\otimes$ to ensure that $\mat{A}(k_1, k_3) \otimes \mat{B}(k_3, k_2) = 0^\mathbb{S}$ whenever $\mat{A}(k_1, k_3) = 0^\mathbb{S}$ or $\mat{B}(k_3, k_2) = 0^\mathbb{S}$, which when combined with the finiteness part of the associative array definition ensures that $\mat{A}(k_1, k_3) \otimes \mat{B}(k_3, k_2) \neq 0^\mathbb{S}$ for at most finitely many values of $k_3$, hence making the iterated binary operator $\bigoplus_{k_3 \in K_2 \cap K_3}$ act only over a finite number of elements, and finally the commutativity of $\oplus$ to account for the lack of an explicit ordering of $K_2 \cap K_3$. 
\footnotetext{To make $\oplus$ well-defined, $\mat{A}(k_1, k_2) \coloneq 0^\mathbb{S}$ whenever $(k_1, k_2) \notin K_1 \times K_2$, and likewise $\mat{B}(k_1, k_2) \coloneq 0^\mathbb{S}$ whenever $(k_1, k_2) \notin K_3 \times K_4$.} 

If $K_1 = \{1, 2, \ldots, N\}$ and $K_2 = \{1, 2, \ldots, M\}$, then an array $\mat{A} \colon K_1 \times K_2 \to \mathbb{S}$ is said to be an $N \times M$ matrix.\footnote{We assume without loss of generality that matrices are $1$-indexed.}

GraphBLAS assumes that the matrices we work with are sparse or even hypersparse. While these do not have formal definitions, for an $N \times M$ matrix $\mat{A}$ we typically describe $\mat{A}$ as \emph{sparse} if $\nnz(\mat{A}) \ll \max(N, M)^2$ and \emph{hypersparse} if $\nnz(\mat{A}) \ll \max(N, M)$. Combining this emphasis on (hyper)sparsity with the recognition that a GraphBLAS semiring does not assume the notion of an additive identity $0^\mathbb{S}$, an $N \times M$ GraphBLAS matrix only stores a set of ordered triples of the form $(k_1, k_2, x) \in \{1, 2, \ldots, N\} \times \{1, 2, \ldots, M\} \times \mathbb{S}$, and in the case that $1 \leq k_3 \leq N$ and $1 \leq k_4 \leq M$ are given for which there is no corresponding stored triple then there is no assumed value the GraphBLAS matrix has at $(k_3, k_4)$. In this way the treatment of arrays and matrices differs from that of GraphBLAS, though in any real implementation the $0^\mathbb{S}$ values of an array or matrix as above would not be stored, so not much generality is lost. To further justify this divergence, in this paper there are numerous situations in which the presense of additive and multiplicative identities is convenient or even necessary.

Given a set $V$, $\powerset(V)$ denotes the set of all subsets of $V$. $V^\star$ denotes the set of all strings with letters drawn from $V$. Given $v_1, v_2, \ldots, v_n \in V$, $\langle v_1, v_2, \ldots, v_n \rangle \in V^\star$ is a string of length $n$. Given two strings $\lambda = \langle v_1, v_2, \ldots, v_n \rangle$ and $\kappa = \langle w_1, w_2, \ldots, w_m \rangle$, their concatenation $\lambda \concat \kappa$ is the string $\langle v_1, v_2, \ldots, v_n, w_1, w_2, \ldots, w_m \rangle$. $\langle \rangle$ denotes the empty string in $V^\star$, the unique string of length $0$. Given sets $A_1, A_2, \ldots, A_n$, $(a_1, a_2, \ldots, a_n) \in A_1 \times A_2 \times \cdots \times A_n$, and $1 \leq k_1 < k_2 < \cdots < k_m \leq n$, we write $\pi_{k_1, k_2, \ldots, k_m}(a_1, a_2, \ldots, a_n) = (a_{k_1}, a_{k_2}, \ldots, a_{k_m})$. Finally, given a function $F \colon A \to B$ and $A' \subseteq A$, then $F[A'] \coloneq \{ F(a) \in B \mid a \in A' \}$.

Unless otherwise specified, we use the following notation: $\mathbb{S}$ denotes the underlying subset of an arbitrary semiring and $x, y, z$ denote elements of $\mathbb{S}$. Bolded, uppercase, upright variables ($\mat{A}, \mat{B}, \mat{C}, \ldots$) denote arrays and matrices. Bolded, lowercase, upright variables ($\mat{x}, \mat{y}, \mat{b}, \ldots$) denote row vectors. $k$ denotes a row or column key while $K$ denotes a set of row or column keys. $\oplus$ denotes the additive operation of a semiring while $\otimes$ denotes the multiplicative operation of a semiring. $M, N, O$ denote dimensions of matrices. $m, n, P$ denote the numbers of elemnts in given sequences of elements. $X, Y, Z$ denote subsets of a given set. $V$ denotes the set of vertices of a graph and $u, v, w$ denote vertices. $E$ denotes the set of edges of a graph and $e, f$ denote edges.

We end this section by reviewing the mathematical ideas underlying the linear algebraic approach to graph algorithms. Given a finite weighted directed graph $(V, E, \mathrm{w} \colon E \to (0, \infty))$, four associative arrays---the adjacency array $\mat{A} \colon V \times V \to \mathbb{S}$, the in-incidence array $\mat{E}_\mathrm{in} \colon V \times E \to \mathbb{S}$, the out-incidence array $\mat{E}_\mathrm{out} \colon V \times E \to \mathbb{S}$, and the edge weight array $\mat{D}_\mathrm{w} \colon E \times E \to \mathbb{S}$---are defined, respectively, by setting:
\begin{align*}
	\mat{A}(u, v) & \coloneq \begin{cases} \mathop{\mathrm{w}}(u, v) & \text{if $(u, v) \in E$,} \\ 0^\mathbb{S} & \text{otherwise,} \end{cases} \\
	\mat{E}_\mathrm{in}(v, e) & \coloneq \begin{cases} 1^\mathbb{S} & \text{if $e = (u, v)$ for some $u \in V$,} \\ 0^\mathbb{S} & \text{otherwise,} \end{cases} \\
	\mat{E}_\mathrm{out}(v, e) & \coloneq \begin{cases} 1^\mathbb{S} & \text{if $e = (v, w)$ for some $w \in V$,} \\ 0^\mathbb{S} & \text{otherwise,} \end{cases} \\
	\mat{D}_\mathrm{w}(e, f) & \coloneq \begin{cases} \mathop{\mathrm{w}}(e) & \text{if $e = f$,} \\ 0^\mathbb{S} & \text{otherwise,} \end{cases}
\end{align*}
where $v, w \in V$ and $e, f \in E$. Note that
\begin{equation*}
	\mat{A} = \mat{E}_\mathrm{out} \arrayprod{\oplus}{\otimes} \mat{D}_\mathrm{w} \arrayprod{\oplus}{\otimes} \mat{E}_\mathrm{in}^\intercal.
\end{equation*}
Graph algorithms are then cast in the language of linear algebra by using the arrays $\mat{A}, \mat{E}_\mathrm{in}, \mat{E}_\mathrm{out}, \mat{D}_\mathrm{w}$. By fixing an order of the vertices $V = \{v_1, v_2, \ldots, v_n\}$ and edges $E = \{e_1, e_2, \ldots, e_m\}$, the arrays $\mat{A}, \mat{E}_\mathrm{in}, \mat{E}_\mathrm{out}, \mat{D}_\mathrm{w}$ can be treated as matrices (of respective dimensions $n \times n$, $n \times m$, $n \times m$, and $m \times m$) which is the standard approach used.


\section{Parallel Algorithms -- Hypersparse Addition Is All You Need}
\label{hypersparse addition}

Dividing matrices, tensors, $n$-dimensional arrays, etc. among different processes using various parallel mapping schemes (block, cyclic, block-cyclic, overlaps, \ldots) is a powerful tool of advanced parallel computing with the inherent advantage of deriving concurrency from data locality. These distributed arrays (also referred to a PGAS---partitioned global address spaces\cite{yelick2007productivity}) often leverage threading and message passing as underlying enabling technologies. Distributed arrays can be implemented effectively in low level languages \cite{de2015partitioned, amarasinghe2023compiler} and high level languages (e.g., Matlab\cite{choy2004star, moler2020history}, Octave\cite{Kepner2009}, Python\cite{byun2023ppython, shajii2023codon}) using both shared and distributed memory hardware. Many message passing libraries present themselves to users as distributed arrays \cite{balay2019petsc, heroux2005overview
}. Distributed arrays align well with high level languages as large arrays are often a core concept within high level languages and operating on large arrays as a whole (vectorization \cite{cai2005performance, birkbeck2007dimension}) is an important optimization technique.

In the context of matrices, distributed arrays divide up the $2$-dimensional space of matrix indices so each processor has blocks of the matrix that it can operate on without inducing interprocessor communication. When communication is required it can be done transparently for the user, hiding the $\mathop{\mathrm{O}}(N_p^2)$ possible communications, where $N_p$ is the number of processors. The essence of parallel partitioning is the local-to-global and global-to-local mappings of matrix indices that allow local computations to leverage locally optimized matrix operations, such as the BLAS (Basic Linear Algebra Subprograms)\cite{lawson1979basic, dongarra1990set, blackford2002updated, barrachina2008evaluation}. For dense matrices, parallel partitioning is a mature technique. When turning attention towards sparse matrices, parallel partitioning raises a different set of considerations. While parallel partitioning of sparse matrices is well understood in certain cases \cite{buluc2008challenges, bulucc2012parallel, merrill2016merge}---most notably structured sparse matrices interacting with dense vectors of the type commonly found in engineering simulations, such as FEM (finite element methods)\cite{law1986parallel, johnsson1990data, yagawa1993parallel}---for hypersparse matrices with randomly distributed non-zero entries, there are abundant opportunities to reexamine how to apply parallel computing techniques. 

Here, we briefly explore the merits of one parallel partition approach, \emph{sum partitioning}, that leverages the unique attributes of hypersparse matrices, whose properties include:
\begin{itemize}
	\item $\nnz(\mat{A}) \ll N$.
	\item Every nonzero entry effectively stores its full $(k_1, k_2, x)$ triple.
	\item $N \times N$ fully covers the desired computational domain.
\end{itemize}
In such a system, it is interesting to consider a simple parallel sum partition of
\begin{equation*}
	\mat{A} = \mat{A}_1 + \mat{A}_2 + \cdots + \mat{A}_p + \cdots + \mat{A}_P = \sum_{p=1}^P{\mat{A}_p},
\end{equation*}
or, more generally,
\begin{equation*}
	\mat{A} = \mat{A}_1 \oplus \mat{A}_2 \oplus \cdots \oplus \mat{A}_p \oplus \cdots \oplus \mat{A}_P = \bigoplus_{p=1}^P{\mat{A}_p}
\end{equation*}
where each processor has a subset of the nonzero entries of the $N \times N$ matrix $\mat{A}$ stored locally in the $N \times N$ matrix $\mat{A}_p$. Sum partitioning can be applied to any matrix at the cost of storing the full triple representation. In the hyperspace regime, the full triple representation is a necessity and so there is no added cost of using sum partitioning. Furthermore, when the nonzero entries are dynamically distributed at random, aggregation operations will nearly always induce full global communication and the benefits of locality based partition schemes are lessened.

A significant benefit of sum partitioning is simplifying parallel algorithm development to the mathematical question of determining how to apply an operation $F$:
\begin{equation*}
	F(\mat{A}) = F(\mat{A}_1 \oplus \mat{A}_2 \oplus \cdots \oplus \mat{A}_P),
\end{equation*}
which can be distilled to separating $F(-)$ into linear and non-linear parts. If $F(-)$ is entirely linear, then
\begin{equation*}
	F(\mat{A}) = F(\mat{A}_1) \oplus F(\mat{A}_2) \oplus \cdots \oplus F(\mat{A}_P)
\end{equation*}
and the parallel computation becomes mapping $F(-)$ on to each $\mat{A}_p$ followed by a $\oplus$ reduction. This presents unique opportunities for merging linear operations and deferring aggregations only when the are required by non-linear operations. This changes the viewpoint of parallel algorithm development from one focused on data partitioning to one focused on managing linear and non-linear operations.

Here the GraphBLAS support of a wide range of linear semiring operations (e.g., the various matrix multiplications $\arrayprod{+}{\times}$, $\arrayprod{\max}{\times}$, $\arrayprod{\min}{+}$, $\arrayprod{\max}{\times}$, $\arrayprod{\min}{\times}$, $\arrayprod{\max}{\min}$, $\arrayprod{\min}{\max}$, \ldots) provides a rich set of tools for writing algorithms in terms of linear semiring operations which obey the necessary linear properties ($\otimes$ distributes over $\oplus$, $\oplus$ is commutative, and $0^\mathbb{S}$ is the additive identity and multiplicative annihilator).  

\noindent \underline{Example 1} Matrix Multiplication. Given any $M \times N$ matrix $\mat{B}$ and $N \times O$ matrix $\mat{C}$, we have
\begin{align*}
	& \mat{B} \arrayprod{\oplus}{\otimes} \mat{A} \arrayprod{\oplus}{\otimes} \mat{C} \\
	& = \mat{B} \arrayprod{\oplus}{\otimes} (\mat{A}_1 \oplus \mat{A}_2 \oplus \cdots \oplus \mat{A}_P) \arrayprod{\oplus}{\otimes} \mat{C} \\
	& = (\mat{B} \arrayprod{\oplus}{\otimes} (\mat{A}_1 \oplus \mat{A}_2 \oplus \cdots \oplus \mat{A}_P)) \arrayprod{\oplus}{\otimes} \mat{C} \\
	& = \left( \bigoplus_{p=1}^P{(\mat{B} \arrayprod{\oplus}{\otimes} \mat{A}_p)} \right) \arrayprod{\oplus}{\otimes} \mat{C} \\
	& = \bigoplus_{p=1}^P{(\mat{B} \arrayprod{\oplus}{\otimes} \mat{A}_p \arrayprod{\oplus}{\otimes} \mat{C})}.
\end{align*}
A particular sub-example is in summing (with respect to $\oplus$) the elements of $\mat{A}$, which can be realized as the product $\mat{1} \arrayprod{\oplus}{\otimes} \mat{A} \arrayprod{\oplus}{\otimes} \mat{1}^\intercal$, where $\mat{1}$ is the $1 \times N$ vector of all $1^\mathbb{S}$s. 

\noindent \underline{Example 2} Element-Wise Multiplication. Another operation that distributes over $\oplus$ is element-wise multiplication denoted by $\otimes$; given an $N \times N$ matrix $\mat{M}$,
\begin{align*}
	& \mat{M} \otimes \mat{A} \\
	& = \mat{M} \otimes (\mat{A}_1 \oplus \mat{A}_2 \oplus \cdots \oplus \mat{A}_P) \\
	& = (\mat{M} \otimes \mat{A}_1) \oplus (\mat{M} \otimes \mat{A}_2) \oplus \cdots \oplus (\mat{M} \otimes \mat{A}_P).
\end{align*}
A concrete application of this would be applying a mask to $\mat{A}$, which can thus be achieved by masking the individual summands $\mat{A}_p$ ($p = 1, 2, \ldots, P$) and then aggregating. 

\noindent \underline{Example 3} Deep Neural Networks (DNN) \cite{clark1955generalization, lippmann1988introduction, lecun2015deep, edelman2024backpropagation}. Consider the case of a $2$-layer ReLU (Rectified Linear Unit)\cite{zeiler2013rectified} feed-forward neural network taking $M$-dimensional inputs $\mat{x} \in \mathbb{R}^M$ and yielding $N$-dimensional outputs $\mat{y} \in \mathbb{R}^N$ by computing
\begin{equation*}
	\mat{y} = \max(\mat{x} \mat{W} + \mat{b}, \mat{0})
\end{equation*}
where $\mat{W}$ is a weight matrix, $\mat{b}$ is a bias vector, and $\mat{0}$ is the $1 \times N$ vector of all $0^\mathbb{S}$s. A parallel sum partition in this context would represent the weight matrix $\mat{W}$ and bias vector $\mat{b}$ as sums $\sum_{p=1}^P{\mat{W}_p}$ and $\sum_{p=1}^P{\mat{b}_p}$, respectively. In such a situation we have
\begin{align*}
	\mat{y} & = \max\left( \mat{x} \left(\sum_{p=1}^P{\mat{W}_p}\right) + \sum_{p=1}^N{\mat{b}_p}, \mat{0} \right) \\
	& = \max\left( \sum_{p=1}^P{(\mat{x} \mat{W}_p + \mat{b}_p)}, \mat{0} \right).
\end{align*}
Although $\max(-, \mat{0})$ is generally \emph{not} linear, in certain situations it can be treated as such. In particular, suppose that there are $N \times N$ $\{0, 1\}$-valued diagonal matrices $\mat{I}_p$ ($p = 1, 2, \ldots, P$) such that $\sum_{p=1}^N{\mat{I}_p}$ is the $N \times N$ identity matrix for matrix multiplication $\mat{I}$, and moreover that $\mat{W}_p = \mat{W} \mat{I}_p$ and $\mat{b}_p = \mat{b} \mat{I}_p$ for each $p = 1, 2, \ldots, P$. Then
\begin{equation*}
	\mat{x} \mat{W}_p + \mat{b}_p = \mat{x} \mat{W} \mat{I}_p  + \mat{b} \mat{I}_p = (\mat{x} \mat{W} + \mat{b}) \mat{I}_p
\end{equation*}
for each $p = 1, 2, \ldots, P$, so given distinct indices $p, q \in \{1, 2, \ldots, P\}$ then we know that $\mat{x} \mat{W}_p + \mat{b}_p$ and $\mat{x} \mat{W}_q + \mat{b}_q$ have no nonzero entries in common. With that observation made, $\max(-, \mat{0})$ \emph{is} linear, so that
\begin{equation*}
	\mat{y} = \max\left( \sum_{p=1}^P{(\mat{x} \mat{W}_p + \mat{b}_p)}, \mat{0} \right) = \sum_{p=1}^P{\max(\mat{x} \mat{W}_p + \mat{b}_p, \mat{0})}.
\end{equation*}
In other words, under such a parallel sum partition the original neural network can be treated as a sum of $P$ disjoint neural networks. 

\noindent \underline{Example 4} Network Traffic Matrices \cite{medina2002traffic, zhang2003fast, tune2013internet}.  For structured decompositions like that of $\mat{W} \mat{I}_p$ above, other operations $F$ also act linearly with respect to $\oplus$. In a traffic matrix $\mat{A}$ the values count the number of packets sent from a source address to a destination address and $\{0, 1\}$-valued diagonal matrices $\mat{I}_p$ select disjoint subsets of the full address space. Aggregate statistics on $\mat{A}$ (e.g., Table~1 in \cite{jananthan2024anonymized}) act linearly with respect to $\oplus$; including the number of unique sources/destinations depending on whether we do source address partitioning: $\sum_{p=1}^P{\mat{I}_p \mat{A}}$,  or destination address partitioning: $\sum_{p=1}^P{\mat{A} \mat{I}_p}$. This also includes cases involving $\max$ as in the DNN example above, such as the maximum number of link packets seen by any source-destination pair.

\section{Mathematical Foundations of Matrix-Based Streaming Graph Algorithms}
\label{matrix-based streaming}

Graphs are abstract mathematical objects with many ways to employ stream processing \cite{feigenbaum2005graph, ediger2012stinger, mcgregor2014graph}. In the context of the GraphBLAS it is useful to focus on matrix-based streaming graph algorithms for which we can draw extensively from the field of digital signal processing (DSP) and is defined in terms of matrices. In DSP the central streaming concept is processing on specific sample window sizes $m$. The number of samples in a window of time $t$ is given by $m = t/dt$, where $f = 1/dt$ is a fixed sampling frequency. DSP is often performed on multiple timescales each with a window size $m_s$. The different windows size are often integer multiples of each other. In particular, the multiples are often powers of $2$ and $3$ if fast Fourier transform (FFT) based methods are to be employed \cite{cochran1967fast, brigham1988fast, duhamel1990fast}.

These DSP concepts map over to matrix-based graph algorithms with the nuance that there is a choice of either fixing $m$ and letting $t$ be variable or fixing $t$ and letting $m$ be variable. In our experience, this choice is determined by the underlying statistical distribution of the graph. If the graph is drawn from a light-tailed distribution, then fixing $t$ is reasonable as the theory for how to adjust light-tail statistics to compensate for variable $m$ is well-established. If the graph is drawn from a heavy-tailed distribution, however, then fixing $m$ works well because less is known on how to adjust heavy-tail statistics to compensate for variable $m$ \cite{nair2020fundamentals}.

Many cases of streaming graphs are drawn from a heavy-tailed distributions \cite{leland1994self, faloutsos1999power, adamic2000power, clauset2009power, barabasi2016network}, which naturally lead to hypersparsity. As such, dense matrix operations can be a challenge, and instead treatment using sparse linear algebra is more effective from a computational perspective. For a sequence of streamed graph adjacency matrices $\ldots, \mat{A}_t, \ldots$ multi-temporal analysis can be done effectively by binary summation. For example, analyzing every $\mat{A}_t$ will focus on phenomena around frequency $\sim f$. Summing consecutive pairs $\mat{A}_t + \mat{A}_{t+dt}$ will be focused around frequency $\sim f/2$. Continuing this process hierarchically enables efficient streaming graph analysis on a wide range of timescales \cite{gadepally2018hyperscaling, kepner19streaming, trigg2022hypersparse}.

A less mentioned assumption of DSP is that at all timescales the sample windows $m_s$ are stored in circular ring buffers. In other words, there is always a timescale at which the data is being dropped and replaced and these ring buffer timescales may be much shorter than standard assumptions on the length of ``age off'' timescales. This is a fundamental choice when thinking about streaming matrix-based graph algorithms. From a practical perspective, we often see streaming data sent to both:
\begin{enumerate}[(1)]
	\item A DSP ring buffer-style processing chain (whose results are stored in a database).
	\item A graph database that is periodically analyzed via queries using extract/transform/load (ETL).
\end{enumerate}
In either case, matrix-based graph algorithms can be used. For (1), GraphBLAS already may be sufficient. For (2), this involves thinking of GraphBLAS as a database, which would require more thought about which concepts from the database and data structures community should be incorporated into GraphBLAS \cite{stonebraker20058, stonebraker2013scidb, besta2023demystifying}. As such, while GraphBLAS can be said to already support DSP-style streaming graph algorithms, database-style streaming graph algorithms remain an open question.

\section{Complex-Index Matrix Graph Algorithms}
\label{complex-index matrices}

\cite{bottcher2024complex} Provides an excellent review of the many interesting uses of complex networks with complex weights. Many of these applications are drawn from physical phenomena with underling complex processes (e.g., quantum mechanics and electromagnetism). This section further develops ideas on how complex index matrices can be useful to graph algorithms. The aim is to provide a mathematical framework for representing such algorithms while fully recognizing that underlying implementations can be achieve with a variety of efficient means that many not actually require storing complex values per se. Let $(V, E, \mathrm{w} \colon E \to (0, \infty))$ be a finite weighted directed graph and $\mat{A}$ its adjacency array.

One scenario that is not immediately expressible in terms of the arrays $\mat{A}, \mat{E}_\mathrm{in}, \mat{E}_\mathrm{out}, \mat{D}_\mathrm{w}$ involves finding vertices and edges. As a concrete example, consider the case where $\mathbb{S}$ is the nonnegative min-plus tropical semiring $([0, \infty] , \min, +, \infty, 0)$. Then $\mat{A}^n(u, v)$ is the least weight of any $n$-hop path from $u$ to $v$, where
\begin{equation*}
	\mat{A}^n \coloneq \underbrace{\mat{A} \arrayprod{\min}{+} \mat{A} \arrayprod{\min}{+} \cdots \arrayprod{\min}{+} \mat{A}}_{\text{$n$ times}}.
\end{equation*}
While $\mat{A}^n(u, v)$ \emph{does} yield the least weight of any $n$-hop path from $u$ to $v$, what it \emph{doesn't} do is provide the paths from $u$ to $v$ which have that least weight. 

In this section we propose a methodology for supporting the computation of vertices and edges via linear algebraic operations (like the aforementioned concrete example) by using a construction similar to that of the field of complex numbers, $\mathbb{C}$, coupled with modified variants of the adjacency and incidence arrays. An example of such a methodology is by redefining the adjacency array $\mat{A}$ by setting $\mat{A}(u, v) \coloneq u + \imaginaryi v$ if there is an edge from $u$ to $v$ and $\mat{A}(u, v) \coloneq 0$ otherwise.

To better motivate our treatment, consider two of the methodologies used to define $\mathbb{C}$. The first is by setting $\mathbb{C} \coloneq \mathbb{R} \times \mathbb{R}$ and then defining, for any $x_1, x_2, y_1, y_2 \in \mathbb{R}$,
\begin{align*}
	(x_1, x_2) + (y_1, y_2) & \coloneq (x_1 + y_1, x_2 + y_2), \\
	(x_1, x_2) \cdot (y_1, y_2) & \coloneq (x_1 y_1 - x_2 y_2, x_1 y_2 + x_2 y_1).
\end{align*}
It can then be checked that $\mathbb{R} \times \mathbb{R}$ with these operations is a field and that $(0, 1) \cdot (0, 1) = (-1, 0)$. The second method is by imagining that a new element $\imaginaryi$ satisfying $\imaginaryi^2 = -1$ has been adjoined to $\mathbb{R}$ to produce a larger set consisting of sums of the form $x + \imaginaryi y$ for $x, y \in \mathbb{R}$.\footnote{In technical terms, this would be realized by taking the polynomial ring $\mathbb{R}[\imaginaryi]$, in which $\imaginaryi$ is an indeterminate (i.e., treated solely as a formal symbol), and quotienting by the maximal ring ideal $\{ \text{polynomials in $\imaginaryi$ divisible by $\imaginaryi^2 + 1$} \}$.}

With this motivation in mind, given a semiring $\mathbb{S}$ we might seek to adjoin a new element $\imaginaryi$ and then consider sums of the form $x + \imaginaryi y$ for $x, y \in \mathbb{S}$. However, an essential aspect of semirings are that they don't necessarily support a notion of negation, so attempting to require something like $\imaginaryi \otimes \imaginaryi$ being the additive inverse of $1^\mathbb{S}$ runs into issues. Both to sidestep these issues and put focus on computations of the form $(x + \imaginaryi 0^\mathbb{S}) \otimes (0^\mathbb{S} + \imaginaryi y)$, we instead require that $\imaginaryi \otimes \imaginaryi = 0^\mathbb{S}$. In other words, $\imaginaryi$ is first treated as a formal symbol without any additional meaning placed upon it, followed by enforcing a new rule for $\imaginaryi$, namely that $\imaginaryi \otimes \imaginaryi = 0^\mathbb{S}$. The specific choice of $\imaginaryi \otimes \imaginaryi = 0^\mathbb{S}$ has no additional meaning but is chosen for simplicity and to start the conversation around this and similar mathematical approaches which expand the expressiveness of linear algebra with regards to graph algorithms.

As such, define the following operations on $\mathbb{S} \times \mathbb{S}$ for any $x_1, x_2, y_1, y_2 \in \mathbb{S}$:
\begin{align*}
	(x_1, x_2) \oplus (y_1, y_2) & \coloneq (x_1 \oplus y_1, x_2 \oplus y_2), \\
	(x_1, x_2) \otimes (y_1, y_2) & \coloneq (x_1 \otimes y_1, (x_1 \otimes y_2) \oplus (x_2 \otimes y_1)).
\end{align*}
We then make the following observations:
\begin{enumerate}[(i)]
	\item $(\mathbb{S} \times \mathbb{S}, \oplus, \otimes, (0^\mathbb{S}, 0^\mathbb{S}), (1^\mathbb{S}, 0^\mathbb{S}))$ is a semiring. (See \hyperref[appendix A]{Appendix A} for the formal proof.)
	\item The map $\mathbb{S} \to \mathbb{S} \times \mathbb{S}$ defined by $x \mapsto (x, 0^\mathbb{S})$ for $x \in \mathbb{S}$ defines a semiring homomorphism from $(\mathbb{S}, \oplus, \otimes, 0^\mathbb{S}, 1^\mathbb{S})$ into $(\mathbb{S} \times \mathbb{S}, \oplus, \otimes, (0^\mathbb{S}, 0^\mathbb{S}), (1^\mathbb{S}, 0^\mathbb{S}))$. In particular, for any $x, y \in \mathbb{S}$ we have $(x, 0^\mathbb{S}) \oplus (y, 0^\mathbb{S}) = (x \oplus y, 0^\mathbb{S})$ and $(x, 0^\mathbb{S}) \otimes (y, 0^\mathbb{S}) = (x \otimes y, 0^\mathbb{S})$.
	\item $(x, y) = (x, 0^\mathbb{S}) \oplus (0^\mathbb{S}, y) = (x, 0^\mathbb{S}) \oplus ((0^\mathbb{S}, 1^\mathbb{S}) \otimes (y, 0^\mathbb{S}))$ for any $x, y \in \mathbb{S}$. By identifying $(x, 0^\mathbb{S})$ and $(y, 0^\mathbb{S})$ with $x$ and $y$, respectively, and writing $\imaginaryi \coloneq (0^\mathbb{S}, 1^\mathbb{S})$, then this justifies the complex notation $(x, y) = x + \imaginaryi y$.
\end{enumerate}
Given $x + \imaginaryi y$, we define:
\begin{align*}
	\real(x + \imaginaryi y) \coloneq x \quad \text{and} \quad 
	\imaginary(x + \imaginaryi y) \coloneq y,
\end{align*}
the \emph{real part} and \emph{imaginary part} of $x + \imaginaryi y$, respectively. 
Moreover, if $\mat{A} \colon V \times V \to \mathbb{S} \times \mathbb{S}$ is an array, then $\real(\mat{A})$ and $\imaginary(\mat{A})$ are arrays $V \times V \to \mathbb{S}$ defined by $\real(\mat{A})(u, v) \coloneq \real(\mat{A}(u, v))$ and $\imaginary(\mat{A})(u, v) \coloneq \imaginary(\mat{A}(u, v))$.

It still remains to show that this construction has utility, so we turn our attention to examples.

\noindent \underline{Example 1} Finding optimal paths between vertices. In other words, our motivating example of finding the $n$-hop paths from one vertex in a finite weighted directed graph to another which have the least weight among all such $n$-hop paths. With $\mathbb{S}$ the nonnegative min-plus tropical semiring, we define a new semiring $(\tilde{\mathbb{S}}, \tilde{\oplus}, \tilde{\otimes}, 0^{\tilde{\mathbb{S}}}, 1^{\tilde{\mathbb{S}}})$ as follows:
\begin{align*}
	\tilde{\mathbb{S}} & \coloneq [0, \infty] \times \powerset(V^\star), \\
	(x, X) \mathbin{\tilde{\oplus}} (y, Y) & \coloneq \begin{cases} (x, X) & \text{if $x < y$,} \\ (y, Y) & \text{if $y < x$,} \\ (x, X \cup Y) & \text{if $x = y$,} \end{cases} \\
	(x, X) \mathbin{\tilde{\otimes}} (y, Y) & \coloneq (x + y, \{\kappa \concat \lambda \mid \kappa \in X, \lambda \in Y\}), \\
	0^{\tilde{\mathbb{S}}} & \coloneq (\infty, \emptyset), \\
	1^{\tilde{\mathbb{S}}} & \coloneq (0, \{\langle\rangle\})
\end{align*}
for $x, y \in [0, \infty]$ and $X, Y \subseteq \powerset(V^\star)$. See \hyperref[appendix A]{Appendix A} for proof of that the semiring axioms hold. Finally, we define $\tilde{\mat{A}} \colon V \times V \to \tilde{\mathbb{S}} \times \tilde{\mathbb{S}}$ by setting
\begin{equation*}
	\tilde{\mat{A}}(u, v) \coloneq (\mat{A}(u, v), \{\langle u, v \rangle\}) + \imaginaryi (\mat{A}(u, v), \{\langle v \rangle\})
\end{equation*}
for any $u, v \in V$.

\begin{prop} \label{paths of optimal weight proposition}
	Recursively define $\mat{B}_n \colon V \times V \to \tilde{\mathbb{S}}$ for $n \geq 1$ by
	\begin{align*}
		\mat{B}_1 & \coloneq \real(\tilde{\mat{A}}), \\
		\mat{B}_{n+1} & \coloneq \mat{B}_n \arrayprod{\tilde{\oplus}}{\tilde{\otimes}} \imaginary(\tilde{\mat{A}}).
	\end{align*}
	Then for any $u, v \in V$, $\mat{B}_n(u, v) = (x, X)$, where $x$ is the least weight among all $n$-hop paths from $u$ to $v$ and $X$ is the set of all $n$-hop paths from $u$ to $v$ with weight $x$.
\end{prop}
\begin{proof}
	See \hyperref[appendix B]{Appendix B} for proof.
\end{proof}

\noindent \underline{Example 2} $\mathrm{CatValMul}$ and $\mathrm{CalKeyMul}$. The array operations of $\mathrm{CatValMul}$ and $\mathrm{CatKeyMul}$ are operations tracking which nonzero values contribute to an entry in an array product $\mat{A} \arrayprod{\oplus}{\otimes} \mat{B}$ and which the indices of the nonzero values contributing to said entries in an array product $\mat{A} \arrayprod{\oplus}{\otimes} \mat{B}$, respectively. These can both be achieved using a newly constructed semiring $(\mathbb{S}', \oplus', \otimes', 0^{\mathbb{S}'}, 1^{\mathbb{S}'})$ defined as follows, where $\mathbb{T} \coloneq (\mathbb{S} \setminus \{0^\mathbb{S}\}) \times (\mathbb{S} \setminus \{0^\mathbb{S}\}) \times (\mathbb{S} \setminus \{0^\mathbb{S}\}) \times V$ and $X, Y \subseteq \mathbb{T}$:
\begin{align*}
	\mathbb{S}' & \coloneq \powerset(\mathbb{T}), \\
	X \oplus' Y & \coloneq X \cup Y, \\
	X \otimes' Y & \coloneq \{ (x_1 \otimes y_1, x_2 \otimes y_2, x_3 \otimes y_3, u) \mid \\
	& \hphantom{{} \coloneq \{} (x_1, x_2, x_3, u) \in X, (y_1, y_2, y_3, u) \in Y \} \cap \mathbb{T} \\
	0^{\mathbb{S}'} & \coloneq \emptyset, \\
	1^{\mathbb{S}'} & \coloneq \{ (1^\mathbb{S}, 1^\mathbb{S}, 1^\mathbb{S}, u) \mid u \in V \}.
\end{align*}
Then, given any $\mat{A} \colon V \times V \to \mathbb{S}$, we define $\mat{A}' \colon V \times V \to \mathbb{S}' \times \mathbb{S}'$ as follows, where $u, v \in V$:
\begin{equation*}
	\mat{A}'(u, v) \coloneq \begin{cases} \{(1^\mathbb{S}, \mat{A}(u, v), \mat{A}(u, v), u)\} & \text{if $\mat{A}(u, v) \neq 0^\mathbb{S}$,} \\ + \imaginaryi \{(\mat{A}(u, v), 1^\mathbb{S}, \mat{A}(u, v), v)\} & \\ \emptyset + \imaginaryi \emptyset & \text{otherwise.} \end{cases}
\end{equation*}

\begin{prop} \label{catvalmul and catvalmul complex computation}
	Suppose $\mat{A}, \mat{B} \colon V \times V \to \mathbb{S}$ are given and let $\mat{C} \coloneq \imaginary(\mat{A}') \arrayprod{\oplus'}{\otimes'} \real(\mat{B}')$. Then for any $u, v \in V$:
	\begin{enumerate}[(a)]
		\item $\mathop{\mathrm{CatValMul}}(\mat{A}, \mat{B})(u, v) = \pi_{1,2}\left[\mat{C}(u, v)\right]$.
		\item $\mathop{\mathrm{CatKeyMul}}(\mat{A}, \mat{B})(u, v) = \pi_4\left[\mat{C}(u, v)\right]$.
	\end{enumerate}
	Moreover,
	\begin{align*}
		(\mat{A} \arrayprod{\oplus}{\otimes} \mat{B})(u, v) & = \bigoplus\left\{ x \otimes y \mid (x, y, z, w) \in \mat{C}(u, v) \right\} \\
		& = \bigoplus\left\{ z \mid (x, y, z, w) \in \mat{C}(u, v) \right\}.
	\end{align*}
\end{prop}
\begin{proof}
	See \hyperref[appendix B]{Appendix B} for proof.
\end{proof}

\section{Conclusion}

    Capabilities beyond enabling graph algorithms to be readily expressed in the language of linear algebra have been a welcome byproduct of the GraphBLAS. These included new performant ways of thinking about algorithms that include leveraging hypersparse matrices for parallel computation, matrix-based graph streaming, and complex-index matrices. Enhancing the GraphBLAS to further take advantage of these new areas is simplified by formalizing these concepts mathematically. This paper formally develops parallel hypersparse matrices, matrix-based graph streaming, and complex-index matrices and illustrates these concepts with various examples to demonstrate their potential merits.  

\section*{Acknowledgment}

The authors wish to acknowledge the following individuals for their contributions and support: LaToya Anderson, William Arcand, Sean Atkins, David Bestor, William Bergeron, Chris Birardi, Bob Bond, Alex Bonn, Daniel Burrill, Chansup Byun,  Timothy Davis, Chris Demchak, Phil Dykstra, Alan Edelman, Peter Fisher, Jeff Gottschalk, Thomas Hardjono, Chris Hill,  Charles Leiserson, Kirsten Malvey, Lauren Milechin, Sanjeev Mohindra, Guillermo Morales, Julie Mullen, Heidi Perry, Sandeep Pisharody, Christian Prothmann, Andrew Prout, Steve Rejto, Albert Reuther, Antonio Rosa, Scott Ruppel, Daniela Rus, Mark Sherman, Scott Weed, Charles Yee, Marc Zissman.

\bibliographystyle{ieeetr}
\bibliography{paper}

\begin{thebibliography}{10}

\bibitem{Mattson2013}
T.~Mattson, D.~Bader, J.~Berry, A.~Buluc, J.~Dongarra, C.~Faloutsos, J.~Feo,
  J.~Gilbert, J.~Gonzalez, B.~Hendrickson, J.~Kepner, C.~Leiseron,
  A.~Lumsdaine, D.~Padua, S.~Poole, S.~Reinhardt, M.~Stonebraker, S.~Wallach,
  and A.~Yoo, ``Standards for graph algorithm primitives,'' in {\em High
  Performance Extreme Computing Conference (HPEC)}, IEEE, 2013.

\bibitem{kepner2011graph}
J.~Kepner and J.~Gilbert, {\em Graph algorithms in the language of linear
  algebra}.
\newblock SIAM, 2011.

\bibitem{kepner16mathematical}
J.~{Kepner}, P.~{Aaltonen}, D.~{Bader}, A.~{Bulu{\c{c}}}, F.~{Franchetti},
  J.~{Gilbert}, D.~{Hutchison}, M.~{Kumar}, A.~{Lumsdaine}, H.~{Meyerhenke},
  S.~{McMillan}, C.~{Yang}, J.~D. {Owens}, M.~{Zalewski}, T.~{Mattson}, and
  J.~{Moreira}, ``Mathematical foundations of the graphblas,'' in {\em 2016
  IEEE High Performance Extreme Computing Conference (HPEC)}, pp.~1--9, 2016.

\bibitem{bulucc2017design}
A.~Bulu{\c{c}}, T.~Mattson, S.~McMillan, J.~Moreira, and C.~Yang, ``{Design of
  the GraphBLAS API for C},'' in {\em Parallel and Distributed Processing
  Symposium Workshops (IPDPSW), 2017 IEEE International}, pp.~643--652, IEEE,
  2017.

\bibitem{brock2021introduction}
B.~Brock, A.~Bulu{\c{c}}, T.~G. Mattson, S.~McMillan, and J.~E. Moreira,
  ``Introduction to graphblas 2.0,'' in {\em 2021 IEEE International Parallel
  and Distributed Processing Symposium Workshops (IPDPSW)}, pp.~253--262, IEEE,
  2021.

\bibitem{kimmerer2024graphblas}
R.~Kimmerer, T.~G. Mattson, S.~McMillan, B.~Brock, E.~Welch, M.~Pelletier, and
  J.~E. Moreira, ``The graphblas 3.0 project,'' in {\em 2024 IEEE International
  Parallel and Distributed Processing Symposium Workshops (IPDPSW)},
  pp.~478--481, IEEE, 2024.

\bibitem{bulucc2011combinatorial}
A.~Bulu{\c{c}} and J.~R. Gilbert, ``The combinatorial blas: Design,
  implementation, and applications,'' {\em The International Journal of High
  Performance Computing Applications}, vol.~25, no.~4, pp.~496--509, 2011.

\bibitem{kepner2012dynamic}
J.~Kepner, W.~Arcand, W.~Bergeron, N.~Bliss, R.~Bond, C.~Byun, G.~Condon,
  K.~Gregson, M.~Hubbell, J.~Kurz, {\em et~al.}, ``Dynamic distributed
  dimensional data model (d4m) database and computation system,'' in {\em 2012
  IEEE International Conference on Acoustics, Speech and Signal Processing
  (ICASSP)}, pp.~5349--5352, IEEE, 2012.

\bibitem{sundaram2015graphmat}
N.~Sundaram, N.~Satish, M.~M.~A. Patwary, S.~R. Dulloor, M.~J. Anderson, S.~G.
  Vadlamudi, D.~Das, and P.~Dubey, ``Graphmat: high performance graph analytics
  made productive,'' {\em Proceedings of the VLDB Endowment}, vol.~8, no.~11,
  pp.~1214--1225, 2015.

\bibitem{hutchison2015graphulo}
D.~Hutchison, J.~Kepner, V.~Gadepally, and A.~Fuchs, ``Graphulo implementation
  of server-side sparse matrix multiply in the accumulo database,'' in {\em
  2015 IEEE High Performance Extreme Computing Conference (HPEC)}, pp.~1--7,
  IEEE, 2015.

\bibitem{zhang2016gbtl}
P.~Zhang, M.~Zalewski, A.~Lumsdaine, S.~Misurda, and S.~McMillan, ``Gbtl-cuda:
  Graph algorithms and primitives for gpus,'' in {\em 2016 IEEE International
  Parallel and Distributed Processing Symposium Workshops (IPDPSW)},
  pp.~912--920, IEEE, 2016.

\bibitem{kumar2018ibm}
M.~Kumar, W.~P. Horn, J.~Kepner, J.~E. Moreira, and P.~Pattnaik, ``Ibm power9
  and cognitive computing,'' {\em IBM Journal of Research and Development},
  vol.~62, no.~4/5, pp.~10--1, 2018.

\bibitem{davis2019algorithm}
T.~A. Davis, ``Algorithm 1000: Suitesparse: Graphblas: Graph algorithms in the
  language of sparse linear algebra,'' {\em ACM Transactions on Mathematical
  Software (TOMS)}, vol.~45, no.~4, pp.~1--25, 2019.

\bibitem{mcmillan2018design}
S.~Mcmillan, ``Design and implementation of the graphblas template library
  (gbtl) v2. 0,'' tech. rep., 2018.

\bibitem{kepner2018mathematics}
J.~Kepner and H.~Jananthan, {\em Mathematics of big data: Spreadsheets,
  databases, matrices, and graphs}.
\newblock MIT Press, 2018.

\bibitem{buluc2009parallel}
A.~Bulu{\c{c}}, J.~T. Fineman, M.~Frigo, J.~R. Gilbert, and C.~E. Leiserson,
  ``Parallel sparse matrix-vector and matrix-transpose-vector multiplication
  using compressed sparse blocks,'' in {\em Proceedings of the twenty-first
  annual symposium on Parallelism in algorithms and architectures},
  pp.~233--244, 2009.

\bibitem{trigg2022hypersparse}
T.~Trigg, C.~Meiners, S.~Pisharody, H.~Jananthan, M.~Jones, A.~Michaleas,
  T.~Davis, E.~Welch, W.~Arcand, D.~Bestor, W.~Bergeron, C.~Byun, V.~Gadepally,
  M.~Houle, M.~Hubbell, A.~Klein, P.~Michaleas, L.~Milechin, J.~Mullen,
  A.~Prout, A.~Reuther, A.~Rosa, S.~Samsi, D.~Stetson, C.~Yee, and J.~Kepner,
  ``Hypersparse network flow analysis of packets with graphblas,'' in {\em 2022
  IEEE High Performance Extreme Computing Conference (HPEC)}, pp.~1--7, 2022.

\bibitem{jones2023deployment}
M.~Jones, J.~Kepner, A.~Prout, T.~Davis, W.~Arcand, D.~Bestor, W.~Bergeron,
  C.~Byun, V.~Gadepally, M.~Houle, M.~Hubbell, H.~Jananthan, A.~Klein,
  L.~Milechin, G.~Morales, J.~Mullen, R.~Patel, S.~Pisharody, A.~Reuther,
  A.~Rosa, S.~Samsi, C.~Yee, and P.~Michaleas, ``Deployment of real-time
  network traffic analysis using graphblas hypersparse matrices and d4m
  associative arrays,'' in {\em 2023 IEEE High Performance Extreme Computing
  Conference (HPEC)}, pp.~1--8, 2023.

\bibitem{jananthan2024anonymized}
H.~Jananthan, M.~Jones, W.~Arcand, D.~Bestor, W.~Bergeron, D.~Burrill,
  A.~Buluc, C.~Byun, T.~Davis, V.~Gadepally, D.~Grant, M.~Houle, M.~Hubbell,
  P.~Luszczek, P.~Michaleas, L.~Milechin, C.~Milner, G.~Morales, A.~Morris,
  J.~Mullen, R.~Patel, A.~Pentland, S.~Pisharody, A.~Prout, A.~Reuther,
  A.~Rosa, G.~Wachman, C.~Yee, and J.~Kepner, ``Anonymized network sensing
  graph challenge,'' in {\em 2024 IEEE High Performance Extreme Computing
  Conference (HPEC)}, pp.~1--8, 2024.

\bibitem{kepner2024what}
J.~Kepner, H.~Jananthan, M.~Jones, W.~Arcand, D.~Bestor, W.~Bergeron,
  D.~Burrill, A.~Buluc, C.~Byun, T.~Davis, V.~Gadepally, D.~Grant, M.~Houle,
  M.~Hubbell, P.~Luszczek, L.~Milechin, C.~Milner, G.~Morales, A.~Morris,
  J.~Mullen, R.~Patel, A.~Pentland, S.~Pisharody, A.~Prout, A.~Reuther,
  A.~Rosa, G.~Wachman, C.~Yee, and P.~Michaleas, ``What is normal? a big data
  observational science model of anonymized internet traffic,'' in {\em 2024
  IEEE High Performance Extreme Computing Conference (HPEC)}, pp.~1--7, 2024.

\bibitem{oppenheim1975digital}
A.~V. Oppenheim and R.~W. Schafer, {\em Digital Signal Processing}.
\newblock Pearson, 1975.

\bibitem{oppenheim1978applications}
A.~V. Oppenheim, ``Applications of digital signal processing,'' {\em Englewood
  Cliffs}, 1978.

\bibitem{hayes1996statistical}
M.~H. Hayes, {\em Statistical digital signal processing and modeling}.
\newblock John Wiley \& Sons, 1996.

\bibitem{smith1997scientist}
S.~W. Smith {\em et~al.}, ``The scientist and engineer's guide to digital
  signal processing,'' 1997.

\bibitem{manolakis2011applied}
D.~G. Manolakis and V.~K. Ingle, {\em Applied digital signal processing: theory
  and practice}.
\newblock Cambridge university press, 2011.

\bibitem{gadepally2018hyperscaling}
V.~{Gadepally}, J.~{Kepner}, L.~{Milechin}, W.~{Arcand}, D.~{Bestor},
  B.~{Bergeron}, C.~{Byun}, M.~{Hubbell}, M.~{Houle}, M.~{Jones},
  P.~{Michaleas}, J.~{Mullen}, A.~{Prout}, A.~{Rosa}, C.~{Yee}, S.~{Samsi}, and
  A.~{Reuther}, ``Hyperscaling internet graph analysis with d4m on the mit
  supercloud,'' in {\em 2018 IEEE High Performance extreme Computing Conference
  (HPEC)}, pp.~1--6, Sep. 2018.

\bibitem{kepner19streaming}
J.~{Kepner}, V.~{Gadepally}, L.~{Milechin}, S.~{Samsi}, W.~{Arcand},
  D.~{Bestor}, W.~{Bergeron}, C.~{Byun}, M.~{Hubbell}, M.~{Houle}, M.~{Jones},
  A.~{Klein}, P.~{Michaleas}, J.~{Mullen}, A.~{Prout}, A.~{Rosa}, C.~{Yee}, and
  A.~{Reuther}, ``Streaming 1.9 billion hypersparse network updates per second
  with d4m,'' in {\em 2019 IEEE High Performance Extreme Computing Conference
  (HPEC)}, pp.~1--6, 2019.

\bibitem{bottcher2024complex}
L.~B{\"o}ttcher and M.~A. Porter, ``Complex networks with complex weights,''
  {\em Physical Review E}, vol.~109, no.~2, p.~024314, 2024.

\bibitem{fu2023algebraic}
E.~Fu, H.~Jananthan, and J.~Kepner, ``Algebraic conditions on one-step
  breadth-first search,'' in {\em 2023 IEEE MIT Undergraduate Research
  Technology Conference (URTC)}, pp.~1--5, 2023.

\bibitem{jananthan2017polystore}
H.~Jananthan, Z.~Zhou, V.~Gadepally, D.~Hutchison, S.~Kim, and J.~Kepner,
  ``Polystore mathematics of relational algebra,'' in {\em 2017 IEEE
  International Conference on Big Data (Big Data)}, pp.~3180--3189, 2017.

\bibitem{jananthan2017constructing}
H.~Jananthan, K.~Dibert, and J.~Kepner, ``Constructing adjacency arrays from
  incidence arrays,'' in {\em 2017 IEEE International Parallel and Distributed
  Processing Symposium Workshops (IPDPSW)}, pp.~608--615, 2017.

\bibitem{kepner2021mathematics}
J.~Kepner, T.~Davis, V.~Gadepally, H.~Jananthan, and L.~Milechin, ``Mathematics
  of digital hyperspace,'' in {\em 2021 IEEE International Parallel and
  Distributed Processing Symposium Workshops (IPDPSW)}, pp.~263--271, 2021.

\bibitem{lee2024eigenvalue}
D.~Lee, H.~Jananthan, and J.~Kepner, ``Eigenvalue distribution of max-plus
  random matrices,'' in {\em 2024 IEEE MIT Undergraduate Research Technology
  Conference (URTC)}, pp.~1--6, 2024.

\bibitem{min2023fuzzy}
K.~Min, H.~Jananthan, and J.~Kepner, ``Fuzzy relational databases via
  associative arrays,'' in {\em 2023 IEEE MIT Undergraduate Research Technology
  Conference (URTC)}, pp.~1--5, 2023.

\bibitem{golan2013semirings}
J.~S. Golan, {\em Semirings and their Applications}.
\newblock Springer Science \& Business Media, 2013.

\bibitem{brock2023graphblas}
B.~Brock, A.~Buluc, R.~Kimmerer, J.~Kitchen, M.~Kumar, T.~Mattson, S.~McMillan,
  J.~Moreira, M.~Pelletier, and E.~Welch, ``The {GraphBLAS} {C} {API}
  specification.'' \url{https://graphblas.org/docs/GraphBLAS_API_C_v2.1.0.pdf},
  2023.

\bibitem{yelick2007productivity}
K.~Yelick, D.~Bonachea, W.-Y. Chen, P.~Colella, K.~Datta, J.~Duell, S.~L.
  Graham, P.~Hargrove, P.~Hilfinger, P.~Husbands, {\em et~al.}, ``Productivity
  and performance using partitioned global address space languages,'' in {\em
  Proceedings of the 2007 international workshop on Parallel symbolic
  computation}, pp.~24--32, 2007.

\bibitem{de2015partitioned}
M.~De~Wael, S.~Marr, B.~De~Fraine, T.~Van~Cutsem, and W.~De~Meuter,
  ``Partitioned global address space languages,'' {\em ACM Computing Surveys
  (CSUR)}, vol.~47, no.~4, pp.~1--27, 2015.

\bibitem{amarasinghe2023compiler}
S.~Amarasinghe, ``Compiler support for structured data,'' in {\em Proceedings
  of the 2023 ACM/SIGDA International Symposium on Field Programmable Gate
  Arrays}, pp.~1--2, 2023.

\bibitem{choy2004star}
R.~Choy, A.~Edelman, J.~R. Gilbert, V.~Shah, and D.~Cheng, ``Star-p: High
  productivity parallel computing,'' in {\em In 8th Annual Workshop on
  High-Performance Embedded Computing (HPEC 04)}, 2004.

\bibitem{moler2020history}
C.~Moler and J.~Little, ``A history of matlab,'' {\em Proceedings of the ACM on
  Programming Languages}, vol.~4, no.~HOPL, pp.~1--67, 2020.

\bibitem{Kepner2009}
J.~Kepner, {\em Parallel MATLAB for Multicore and Multinode Computers}.
\newblock SIAM, 2009.

\bibitem{byun2023ppython}
C.~Byun, W.~Arcand, D.~Bestor, B.~Bergeron, V.~Gadepally, M.~Houle, M.~Hubbell,
  H.~Jananthan, M.~Jones, A.~Klein, {\em et~al.}, ``{pPython} performance
  study,'' in {\em 2023 IEEE High Performance Extreme Computing Conference
  (HPEC)}, pp.~1--7, IEEE, 2023.

\bibitem{shajii2023codon}
A.~Shajii, G.~Ramirez, H.~Smajlovi{\'c}, J.~Ray, B.~Berger, S.~Amarasinghe, and
  I.~Numanagi{\'c}, ``Codon: A compiler for high-performance pythonic
  applications and dsls,'' in {\em Proceedings of the 32nd ACM SIGPLAN
  International Conference on Compiler Construction}, pp.~191--202, 2023.

\bibitem{balay2019petsc}
S.~Balay, S.~Abhyankar, M.~Adams, J.~Brown, P.~Brune, K.~Buschelman, L.~Dalcin,
  A.~Dener, V.~Eijkhout, W.~Gropp, {\em et~al.}, ``{PETSc} users manual,''
  2019.

\bibitem{heroux2005overview}
M.~A. Heroux, R.~A. Bartlett, V.~E. Howle, R.~J. Hoekstra, J.~J. Hu, T.~G.
  Kolda, R.~B. Lehoucq, K.~R. Long, R.~P. Pawlowski, E.~T. Phipps, {\em
  et~al.}, ``An overview of the {Trilinos} project,'' {\em ACM Transactions on
  Mathematical Software (TOMS)}, vol.~31, no.~3, pp.~397--423, 2005.

\bibitem{cai2005performance}
X.~Cai, H.~P. Langtangen, and H.~Moe, ``On the performance of the {Python}
  programming language for serial and parallel scientific computations,'' {\em
  Scientific Programming}, vol.~13, no.~1, pp.~31--56, 2005.

\bibitem{birkbeck2007dimension}
N.~Birkbeck, J.~Levesque, and J.~N. Amaral, ``A dimension abstraction approach
  to vectorization in matlab,'' in {\em International Symposium on Code
  Generation and Optimization (CGO'07)}, pp.~115--130, IEEE, 2007.

\bibitem{lawson1979basic}
C.~L. Lawson, R.~J. Hanson, D.~R. Kincaid, and F.~T. Krogh, ``Basic linear
  algebra subprograms for fortran usage,'' {\em ACM Transactions on
  Mathematical Software (TOMS)}, vol.~5, no.~3, pp.~308--323, 1979.

\bibitem{dongarra1990set}
J.~J. Dongarra, J.~Du~Croz, S.~Hammarling, and I.~S. Duff, ``A set of level 3
  basic linear algebra subprograms,'' {\em ACM Transactions on Mathematical
  Software (TOMS)}, vol.~16, no.~1, pp.~1--17, 1990.

\bibitem{blackford2002updated}
L.~S. Blackford, A.~Petitet, R.~Pozo, K.~Remington, R.~C. Whaley, J.~Demmel,
  J.~Dongarra, I.~Duff, S.~Hammarling, G.~Henry, {\em et~al.}, ``An updated set
  of basic linear algebra subprograms (blas),'' {\em ACM Transactions on
  Mathematical Software}, vol.~28, no.~2, pp.~135--151, 2002.

\bibitem{barrachina2008evaluation}
S.~Barrachina, M.~Castillo, F.~D. Igual, R.~Mayo, and E.~S. Quintana-Orti,
  ``Evaluation and tuning of the level 3 cublas for graphics processors,'' in
  {\em 2008 IEEE International Symposium on Parallel and Distributed
  Processing}, pp.~1--8, IEEE, 2008.

\bibitem{buluc2008challenges}
A.~Buluc and J.~R. Gilbert, ``Challenges and advances in parallel sparse
  matrix-matrix multiplication,'' in {\em 2008 37th International Conference on
  Parallel Processing}, pp.~503--510, IEEE, 2008.

\bibitem{bulucc2012parallel}
A.~Bulu{\c{c}} and J.~R. Gilbert, ``Parallel sparse matrix-matrix
  multiplication and indexing: Implementation and experiments,'' {\em SIAM
  Journal on Scientific Computing}, vol.~34, no.~4, pp.~C170--C191, 2012.

\bibitem{merrill2016merge}
D.~Merrill and M.~Garland, ``Merge-based parallel sparse matrix-vector
  multiplication,'' in {\em SC'16: Proceedings of the International Conference
  for High Performance Computing, Networking, Storage and Analysis},
  pp.~678--689, IEEE, 2016.

\bibitem{law1986parallel}
K.~H. Law, ``A parallel finite element solution method,'' {\em Computers \&
  Structures}, vol.~23, no.~6, pp.~845--858, 1986.

\bibitem{johnsson1990data}
S.~L. Johnsson and K.~K. Mathur, ``Data structures and algorithms for the
  finite element method on a data parallel supercomputer,'' {\em International
  Journal for Numerical Methods in Engineering}, vol.~29, no.~4, pp.~881--908,
  1990.

\bibitem{yagawa1993parallel}
G.~Yagawa, A.~Yoshioka, S.~Yoshimura, and N.~Soneda, ``A parallel finite
  element method with a supercomputer network,'' {\em Computers \& Structures},
  vol.~47, no.~3, pp.~407--418, 1993.

\bibitem{clark1955generalization}
W.~A. Clark and B.~G. Farley, ``Generalization of pattern recognition in a
  self-organizing system,'' in {\em Proceedings of the March 1-3, 1955, western
  joint computer conference}, pp.~86--91, 1955.

\bibitem{lippmann1988introduction}
R.~P. Lippmann, ``An introduction to computing with neural nets,'' {\em ACM
  SIGARCH Computer Architecture News}, vol.~16, no.~1, pp.~7--25, 1988.

\bibitem{lecun2015deep}
Y.~LeCun, Y.~Bengio, and G.~Hinton, ``Deep learning,'' {\em nature}, vol.~521,
  no.~7553, pp.~436--444, 2015.

\bibitem{edelman2024backpropagation}
A.~Edelman, E.~Aky{\"u}rek, and Y.~Wang, ``Backpropagation through back
  substitution with a backslash,'' {\em SIAM Journal on Matrix Analysis and
  Applications}, vol.~45, no.~1, pp.~429--449, 2024.

\bibitem{zeiler2013rectified}
M.~D. Zeiler, M.~Ranzato, R.~Monga, M.~Mao, K.~Yang, Q.~V. Le, P.~Nguyen,
  A.~Senior, V.~Vanhoucke, J.~Dean, {\em et~al.}, ``On rectified linear units
  for speech processing,'' in {\em 2013 IEEE International Conference on
  Acoustics, Speech and Signal Processing}, pp.~3517--3521, IEEE, 2013.

\bibitem{medina2002traffic}
A.~Medina, N.~Taft, K.~Salamatian, S.~Bhattacharyya, and C.~Diot, ``Traffic
  matrix estimation: Existing techniques and new directions,'' {\em ACM SIGCOMM
  Computer Communication Review}, vol.~32, no.~4, pp.~161--174, 2002.

\bibitem{zhang2003fast}
Y.~Zhang, M.~Roughan, N.~Duffield, and A.~Greenberg, ``Fast accurate
  computation of large-scale ip traffic matrices from link loads,'' {\em ACM
  SIGMETRICS Performance Evaluation Review}, vol.~31, no.~1, pp.~206--217,
  2003.

\bibitem{tune2013internet}
P.~Tune, M.~Roughan, H.~Haddadi, and O.~Bonaventure, ``Internet traffic
  matrices: A primer,'' {\em Recent Advances in Networking}, vol.~1, pp.~1--56,
  2013.

\bibitem{feigenbaum2005graph}
J.~Feigenbaum, S.~Kannan, A.~McGregor, S.~Suri, and J.~Zhang, ``On graph
  problems in a semi-streaming model,'' {\em Theoretical Computer Science},
  vol.~348, no.~2-3, pp.~207--216, 2005.

\bibitem{ediger2012stinger}
D.~Ediger, R.~McColl, J.~Riedy, and D.~A. Bader, ``Stinger: High performance
  data structure for streaming graphs,'' in {\em 2012 IEEE Conference on High
  Performance Extreme Computing}, pp.~1--5, IEEE, 2012.

\bibitem{mcgregor2014graph}
A.~McGregor, ``Graph stream algorithms: a survey,'' {\em ACM SIGMOD Record},
  vol.~43, no.~1, pp.~9--20, 2014.

\bibitem{cochran1967fast}
W.~T. Cochran, J.~W. Cooley, D.~L. Favin, H.~D. Helms, R.~A. Kaenel, W.~W.
  Lang, G.~C. Maling, D.~E. Nelson, C.~M. Rader, and P.~D. Welch, ``What is the
  fast fourier transform?,'' {\em Proceedings of the IEEE}, vol.~55, no.~10,
  pp.~1664--1674, 1967.

\bibitem{brigham1988fast}
E.~O. Brigham, {\em The fast Fourier transform and its applications}.
\newblock Prentice-Hall, Inc., 1988.

\bibitem{duhamel1990fast}
P.~Duhamel and M.~Vetterli, ``Fast fourier transforms: a tutorial review and a
  state of the art,'' {\em Signal processing}, vol.~19, no.~4, pp.~259--299,
  1990.

\bibitem{nair2020fundamentals}
J.~Nair, A.~Wierman, and B.~Zwart, ``The fundamentals of heavy tails:
  Properties, emergence, and estimation,'' {\em Preprint, California Institute
  of Technology}, 2020.

\bibitem{leland1994self}
W.~E. Leland, M.~S. Taqqu, W.~Willinger, and D.~V. Wilson, ``On the
  self-similar nature of ethernet traffic (extended version),'' {\em IEEE/ACM
  Transactions on Networking (ToN)}, vol.~2, no.~1, pp.~1--15, 1994.

\bibitem{faloutsos1999power}
M.~Faloutsos, P.~Faloutsos, and C.~Faloutsos, ``On power-law relationships of
  the internet topology,'' in {\em ACM SIGCOMM computer communication review},
  vol.~29-4, pp.~251--262, ACM, 1999.

\bibitem{adamic2000power}
L.~A. Adamic and B.~A. Huberman, ``Power-law distribution of the world wide
  web,'' {\em science}, vol.~287, no.~5461, pp.~2115--2115, 2000.

\bibitem{clauset2009power}
A.~Clauset, C.~R. Shalizi, and M.~E. Newman, ``Power-law distributions in
  empirical data,'' {\em SIAM review}, vol.~51, no.~4, pp.~661--703, 2009.

\bibitem{barabasi2016network}
A.-L. Barab{\'a}si {\em et~al.}, {\em Network science}.
\newblock Cambridge university press, 2016.

\bibitem{stonebraker20058}
M.~Stonebraker, U.~{\c{C}}etintemel, and S.~Zdonik, ``The 8 requirements of
  real-time stream processing,'' {\em ACM Sigmod Record}, vol.~34, no.~4,
  pp.~42--47, 2005.

\bibitem{stonebraker2013scidb}
M.~Stonebraker, P.~Brown, D.~Zhang, and J.~Becla, ``Scidb: A database
  management system for applications with complex analytics,'' {\em Computing
  in Science \& Engineering}, vol.~15, no.~3, pp.~54--62, 2013.

\bibitem{besta2023demystifying}
M.~Besta, R.~Gerstenberger, E.~Peter, M.~Fischer, M.~Podstawski, C.~Barthels,
  G.~Alonso, and T.~Hoefler, ``Demystifying graph databases: Analysis and
  taxonomy of data organization, system designs, and graph queries,'' {\em ACM
  Computing Surveys}, vol.~56, no.~2, pp.~1--40, 2023.

\end{thebibliography}

\section*{Appendix A: Formal Verification of Semiring Axioms}
\label{appendix A}

The first algebra we prove is a semiring is that of $(\mathbb{S} \times \mathbb{S}, \oplus, \otimes, (0^\mathbb{S}, 0^\mathbb{S}), (1^\mathbb{S}, 0^\mathbb{S}))$, where $(\mathbb{S}, \oplus, \otimes, 0^\mathbb{S}, 1^\mathbb{S})$ is a given semiring. As defined, the additive structure $(\mathbb{S} \times \mathbb{S}, \oplus, (0^\mathbb{S}, 0^\mathbb{S}))$ is equal to the product of the additive commutative monoid $(\mathbb{S}, \oplus, 0^\mathbb{S})$ with itself, and hence is a commutative monoid itself.\footnote{Products of algebras preserve satisfaction of equational laws. As commutative monoids are defined solely by equational laws, this implies that products of commutative monoids are commutative monoids.}

It remains to show that the multiplicative structure $(\mathbb{S} \times \mathbb{S}, \otimes, (1^\mathbb{S}, 0^\mathbb{S}))$ forms a monoid, that $(0^\mathbb{S}, 0^\mathbb{S})$ is an annihilator for $\otimes$, and that $\otimes$ distributes over $\oplus$. To that effect, suppose $x_1, x_2, y_1, y_2, z_1, z_2 \in \mathbb{S}$ are arbitrary.
\begin{itemize}
	\item Associativity of $\otimes$:
	\begin{align*}
		& \bigl( (x_1, x_2) \otimes (y_1, y_2) \bigr) \otimes (z_1, z_2) \\
		& = \bigl( x_1 \otimes y_1, (x_2 \otimes y_1) \oplus (x_1 \otimes y_2) \bigr) \otimes (z_1, z_2) \\
		& = \Bigl( x_1 \otimes y_1 \otimes z_1, \\
		& \qquad (x_1 \otimes y_1 \otimes z_2) \oplus \Bigl( \bigl( (x_2 \otimes y_1) \oplus (x_1 \otimes y_2) \bigr) \otimes z_1 \Bigr) \Bigr) \\
		& = \bigl( x_1 \otimes y_1 \otimes z_1, \\
		& \qquad (x_1 \otimes y_1 \otimes z_2) \oplus (x_2 \otimes y_1 \otimes z_1) \oplus (x_1 \otimes y_2 \otimes z_1) \bigr), \\
		& (x_1, x_2) \otimes \bigl( (y_1, y_2) \otimes (z_1, z_2) \bigr) \\
		& = (x_1, x_2) \otimes \bigl( y_1 \otimes z_1, (y_1 \otimes z_2) \oplus (y_2 \otimes z_1) \bigr) \\
		& = \Bigl( x_1 \otimes y_1 \otimes z_1, \\
		& \qquad \Bigl( x_1 \otimes \bigl( (y_1 \otimes z_2) \oplus (y_2 \otimes z_1) \bigr) \Bigr) \oplus (x_2 \otimes y_1 \otimes z_1) \Bigr) \\
		& = \bigl( x_1 \otimes y_1 \otimes z_1, \\
		& \qquad (x_1 \otimes y_1 \otimes z_2) \oplus (x_1 \otimes y_2 \otimes z_1) \oplus (x_2 \otimes y_1 \otimes z_1) \bigr).
	\end{align*}
	
	\item $(1^\mathbb{S}, 0^\mathbb{S})$ is an identity for $\otimes$:
	\begin{align*}
		(1^\mathbb{S}, 0^\mathbb{S}) \otimes (x, y) & = \bigl( 1^\mathbb{S} \otimes x, (0^\mathbb{S} \otimes x) \oplus (1^\mathbb{S} \otimes y) \bigr) \\
		& = (x, 0^\mathbb{S} \oplus y) \\
		& = (x, y), \\
		(x, y) \otimes (1^\mathbb{S}, 0^\mathbb{S}) & = \bigl( x \otimes 1^\mathbb{S}, (x \otimes 0^\mathbb{S}) \oplus (y \otimes 1^\mathbb{S}) \bigr) \\
		& = (x, 0^\mathbb{S} \oplus y) \\
		& = (x, y).
	\end{align*}
	
	\item $(0^\mathbb{S}, 0^\mathbb{S})$ is an annihilator for $\otimes$:
	\begin{align*}
		(0^\mathbb{S}, 0^\mathbb{S}) \otimes (x, y) & = \bigl( 0^\mathbb{S} \otimes x, (0^\mathbb{S} \otimes y) \oplus (0^\mathbb{S} \otimes x) \bigr) \\
		& = (0^\mathbb{S}, 0^\mathbb{S} \oplus 0^\mathbb{S}) \\
		& = (0^\mathbb{S}, 0^\mathbb{S}), \\
		(x, y) \otimes (0^\mathbb{S}, 0^\mathbb{S}) & = \bigl( x \otimes 0^\mathbb{S}, (x \otimes 0^\mathbb{S}) \oplus (y \otimes 0^\mathbb{S}) \bigr) \\
		& = (0^\mathbb{S}, 0^\mathbb{S} \oplus 0^\mathbb{S}) \\
		& = (0^\mathbb{S}, 0^\mathbb{S}).
	\end{align*}
	
	\item Distributivity of $\otimes$ over $\oplus$:
	\begin{align*}
		& (x_1, x_2) \otimes \bigl( (y_1, y_2) \oplus (z_1, z_2) \bigr) \\
		& = (x_1, x_2) \otimes (y_1 \oplus z_1, y_2 \oplus z_2) \\
		& = \Bigl( x_1 \otimes (y_1 \oplus z_1), \\
		& \hphantom{{} = \Bigl(} \bigl( x_2 \otimes (y_1 \oplus z_1) \bigr) \oplus \bigl( x_1 \otimes (y_2 \oplus z_2) \bigr) \Bigr) \\
		& = \bigl( (x_1 \otimes y_1) \oplus (x_1 \otimes z_1), \\
		& \hphantom{{} = \Bigl(} (x_2 \otimes y_1) \oplus (x_2 \otimes z_1) \oplus (x_1 \otimes y_2) \oplus (x_1 \otimes z_2) \bigr) \\
		& = \bigl( x_1 \otimes y_1, (x_2 \otimes y_1) \oplus (x_1 \otimes y_2) \bigr) \\
		& \hphantom{{} = } \oplus \bigl( x_1 \otimes z_1, (x_2 \otimes z_1) \oplus (x_1 \otimes z_2) \bigr) \\
		& = \bigl( (x_1, x_2) \otimes (y_1, y_2) \bigr) \oplus \bigl( (x_1, x_2) \otimes (z_1, z_2) \bigr).
	\end{align*}
	Showing $\bigl( (y_1, y_2) \oplus (z_1, z_2) \bigr) \otimes (x_1, x_2) = \bigl( (y_1, y_2) \otimes (x_1, x_2) \bigr) \oplus \bigl( (z_1, z_2) \otimes (x_1, x_2) \bigr)$ is analogous.
\end{itemize}

We next show that $(\tilde{\mathbb{S}}, \tilde{\oplus}, \tilde{\otimes}, 0^{\tilde{\mathbb{S}}}, 1^{\tilde{\mathbb{S}}})$ is a semiring. Given $x, y \in [0, \infty]$ and $X \subseteq \powerset(V^\star)$, let $\delta_{x, y}(X) \coloneq X$ when $x = y$ and $\delta_{x, y}(X) \coloneq \emptyset$ otherwise, so we may write
\begin{equation*} 
	(x, X) \mathbin{\tilde{\oplus}} (y, Y) = (\min(x, y), \delta_{x, \min(x, y)}(X) \cup \delta_{y, \min(x, y)}(Y)).
\end{equation*}
Additionally define
\begin{equation*}
	X \concat Y \coloneq \{ \kappa \concat \lambda \mid \kappa \in X, \lambda \in Y \},
\end{equation*}
so that $(x, X) \mathbin{\tilde{\otimes}} (y, Y) = (x + y, X \concat Y)$.

\begin{lem} \label{tilde oplus lemma}
	Suppose $x, y, z \in [0, \infty]$ and $X, Y, Z \subseteq \powerset(V^\star)$. Then:
	\begin{enumerate}[(x)]
		\item $\delta_{x, y}(X \cup Y) = \delta_{x, y}(X) \cup \delta_{x, y}(Y)$.
		\item If $x \leq y \leq z$, then $\delta_{y, x}(\delta_{z, y}(X)) = \delta_{z, x}(X)$.
		\item If $f \colon [0, \infty] \to [0, \infty]$ is injective, then $\delta_{x, y}(X) = \delta_{f(x), f(y)}(X)$.
		\item $\emptyset \concat X = X \concat \emptyset = \emptyset$.
		\item $\{\langle\rangle\} \concat X = X \concat \{\langle\rangle\} = X$.
		\item $X \concat \delta_{x, y}(Y) = \delta_{x, y}(X \concat Y) = \delta_{x, y}(X) \concat Y$. 
		\item $X \concat (Y \concat Z) = (X \concat Y) \concat Z$.
		\item $X \concat (Y \cup Z) = (X \concat Y) \cup (X \concat Z)$ and $(Y \cup Z) \concat X = (Y \concat X) \cup (Z \concat X)$.
	\end{enumerate}
\end{lem}
\begin{proof} \leavevmode
\begin{enumerate}[(x)]
	\item If $x = y$, then $\delta_{x, y}(X \cup Y) = X \cup Y = \delta_{x, y}(X) \cup \delta_{x, y}(Y)$. Otherwise, $\delta_{x, y}(X \cup Y) = \emptyset = \emptyset \cup \emptyset = \delta_{x, y}(X) \cup \delta_{x, y}(Y)$.
	\item If $x = z$, then $x \leq y \leq z$ implies $x = y = z$. As such, $\delta_{y, x}(\delta_{z, y}(X)) = \delta_{y, x}(X) = X = \delta_{z, x}(X)$. On the other hand, if $x \neq z$, then $x \leq y \leq z$ implies either $x \neq y$ or $y \neq z$; in the former case we have $\delta_{y, x}(\delta_{z, y}(X)) = \emptyset = \delta_{z, x}(X)$ and in the latter case we have $\delta_{y, x}(\delta_{z, y}(X)) = \delta_{y, x}(\emptyset) = \emptyset = \delta_{z, x}(X)$. 
	\item Since $f$ is injective, $x = y$ if and only if $f(x) = f(y)$, so $\delta_{x, y}(-) = \delta_{f(x), f(y)}(-)$.
	\item $\emptyset \concat X = \{ \kappa \concat \lambda \mid \kappa \in \emptyset, \lambda \in X \}$. Since there are no such $\kappa$ in $\emptyset$, this shows that $\emptyset \concat X = \emptyset$. Showing $X \concat \emptyset = \emptyset$ is analogous.
	\item $\{\langle\rangle\} \concat X = \{ \kappa \concat \lambda \mid \kappa \in \{\langle\rangle\}, \lambda \in X \} = \{ \langle\rangle \concat \lambda \mid \lambda \in X \} = \{\lambda \mid \lambda \in X\} = X$. Showing $X \concat \{\langle\rangle\} = X$ is analogous.
	\item If $x = y$, then $X \concat \delta_{x, y}(Y) = X \concat Y = \delta_{x, y}(X \concat Y)$. If $x \neq y$, then $X \concat \delta_{x, y}(Y) = X \concat \emptyset = \emptyset = \delta_{x, y}(X \concat Y)$ by using (d) above. Showing $\delta_{x, y}(X) \concat Y = \delta_{x, y}(X \concat Y)$ is analogous.
	\item Concatenation is associative, so 
	\begin{align*}
		X \concat (Y \concat Z) & = \{ \kappa \concat \lambda \mid \kappa \in X, \lambda \in Y \concat Z \} \\
		& = \{ \kappa \concat (\lambda \concat \mu) \mid \kappa \in X, \lambda \in Y, \mu \in Z \} \\
		& = \{ (\kappa \concat \lambda) \concat \mu \mid \kappa \in X, \lambda \in Y, \mu \in Z \} \\
		& = \{ \kappa \concat \mu \mid \kappa \in X \concat Y, \mu \in Z \} \\
		& = (X \concat Y) \concat Z.
	\end{align*}
	\item As $\wedge$ (logical AND) distributes over $\vee$ (logical OR)
	\begin{align*}
		& X \concat (Y \cup Z) \\
		& = \{ \kappa \concat \lambda \mid \kappa \in X, \lambda \in Y \cup Z \} \\
		& = \{ \kappa \concat \lambda \mid \kappa \in X \wedge (\lambda \in Y \vee \lambda \in Z) \} \\
		& = \{ \kappa \concat \lambda \mid (\kappa \in X \wedge \lambda \in Y) \vee (\kappa \in X \wedge \lambda \in Z) \} \\
		& = \{ \kappa \concat \lambda \mid \kappa \in X, \lambda \in Y \} \cup \{ \kappa \concat \lambda \mid \kappa \in X, \lambda \in Z \} \\
		& = (X \concat Y) \cup (X \concat Z).
	\end{align*}
	Showing $(Y \cup Z) \concat X = (Y \concat X) \cup (Z \concat X)$ is analogous.
\end{enumerate}
\end{proof}

Now suppose $x, y, z \in [0, \infty]$ and $X, Y, Z \subseteq \powerset(V^\star)$ are arbitrary.
\begin{itemize}
	\item Commutativity of $\tilde{\oplus}$: Both $\min$ and $\cup$ are commutative, hence
	\begin{align*}
		& (x, X) \mathbin{\tilde{\oplus}} (y, Y) \\
		& = \bigl( \min(x, y), \delta_{x, \min(x, y)}(X) \cup \delta_{y, \min(x, y)}(Y) \bigr) \\
		& = \bigl( \min(y, x), \delta_{y, \min(y, x)}(Y) \cup \delta_{x, \min(y, x)}(X) \bigr) \\
		& = (y, Y) \mathbin{\tilde{\oplus}} (x, X).
	\end{align*}
	
	\item Associativity of $\tilde{\oplus}$: Noting that $\min(x, y, z) \leq \min(y, z) \leq y, z$, applying Lemma~\ref{tilde oplus lemma}(x, y) shows:
	\begin{align*}
		& (x, X) \mathbin{\tilde{\oplus}} \bigl( (y, Y) \mathbin{\tilde{\oplus}} (z, Z) \bigr) \\
		& = (x, X) \mathbin{\tilde{\oplus}} \bigl( \min(y, z), \delta_{y, \min(y, z)}(Y) \cup \delta_{z, \min(y, z)}(Z) \bigr) \\
		& = \bigl( \min(x, \min(y, z)), \\
		& \hphantom{{} = \bigl(} \delta_{x, \min(x, \min(y, z))}(X) \\
		& \hphantom{{} = \bigl(} \cup \delta_{\min(y, z), \min(x, \min(y, z))}(\delta_{y, \min(y, z)}(Y) \\
		& \hphantom{ \hphantom{{} = \bigl(} \cup \delta_{\min(y, z), \min(x, \min(y, z))}( } \cup \delta_{z, \min(y, z)}(Z)) \bigr) \\
		& = \bigl( \min(x, y, z), \\
		& \hphantom{{} = \bigl(} \delta_{x, \min(x, y, z)}(X) \\
		& \hphantom{{} = \bigl(} \cup \delta_{\min(y, z), \min(x, y, z)}(\delta_{y, \min(y, z)}(Y)) \\
		& \hphantom{{} = \bigl(} \cup \delta_{\min(y, z), \min(x, y, z)}(\delta_{z, \min(y, z)}(Z)) \bigr) \\
		& = \bigl( \min(x, y, z), \\
		& \hphantom{{} = \bigl(} \delta_{x, \min(x, y, z)}(X) \cup \delta_{y, \min(x, y, z)}(Y) \cup \delta_{z, \min(x, y, z)}(Z) \bigr)
	\end{align*}
	As this final expression is symmetric in $x$, $y$, and $z$, using the commutativity of $\tilde{\oplus}$ shown above yields
	\begin{align*}
		& (x, X) \mathbin{\tilde{\oplus}} \bigl( (y, Y) \mathbin{\tilde{\oplus}} (z, Z) \bigr) \\
		& = \bigl( \min(x, y, z), \\
		& \hphantom{{} = \bigl(} \delta_{x, \min(x, y, z)}(X) \cup \delta_{y, \min(x, y, z)}(Y) \cup \delta_{z, \min(x, y, z)}(Z) \bigr) \\
		& = \bigl( \min(z, x, y), \\
		& \hphantom{{} = \bigl(} \delta_{z, \min(z, x, y)}(Z) \cup \delta_{x, \min(z, x, y)}(X) \cup \delta_{y, \min(z, x, y)}(Y) \bigr) \\
		& = (z, Z) \mathbin{\tilde{\oplus}} \bigl( (x, X) \mathbin{\tilde{\oplus}} (y, Y) \bigr) \\
		& = \bigl( (x, X) \mathbin{\tilde{\oplus}} (y, Y) \bigr) \mathbin{\tilde{\oplus}} (z, Z).
	\end{align*}
	
	\item Associativity of $\tilde{\otimes}$: Using the associativity of $+$ along with Lemma~\ref{tilde oplus lemma}(g) shows
	\begin{align*}
		& (x, X) \mathbin{\tilde{\otimes}} \bigl( (y, Y) \mathbin{\tilde{\oplus}} (z, Z) \bigr) \\
		& = (x, X) \mathbin{\tilde{\otimes}} (y + z, Y \concat Z) \\
		& = \bigl( x + (y + z), X \concat (Y \concat Z) \bigr) \\
		& = \bigl( (x + y) + z, (X \concat Y) \concat Z \bigr) \\
		& = (x + y, X \concat Y) \mathbin{\tilde{\otimes}} (z, Z) \\
		& = \bigl( (x, X) \mathbin{\tilde{\otimes}} (y, Y) \bigr) \mathbin{\tilde{\otimes}} (z, Z).
	\end{align*}
	
	\item $0^{\tilde{\mathbb{S}}} \coloneq (\infty, \emptyset)$ is an identity for $\tilde{\oplus}$:
	\begin{align*}
		& (x, X) \mathbin{\tilde{\oplus}} (\infty, \emptyset) \\
		& = \bigl( \min(x, \infty), \delta_{x, \min(x, \infty)}(X) \cup \delta_{\infty, \min(x, \infty)}(\emptyset) \bigr) \\
		&  = (x, \delta_{x, x}(X) \cup \emptyset) \\
		& = (x, X)
	\end{align*}
	
	\item $1^{\tilde{\mathbb{S}}} \coloneq (0, \{\langle\rangle\})$ is an identity for $\tilde{\otimes}$: Using Lemma~\ref{tilde oplus lemma}(e) shows
	\begin{equation*}
		(x, X) \mathbin{\tilde{\otimes}} (0, \{\langle\rangle\}) = ( x + 0, \{\langle\rangle\} \concat X ) = (x, X).
	\end{equation*}
	Showing $(0, \{\langle\rangle\}) \mathbin{\tilde{\otimes}} (x, X) = (x, X)$ is analogous.
	
	\item $0^{\tilde{\mathbb{S}}} \coloneq (\infty, \emptyset)$ is an annihilator for $\tilde{\otimes}$: Using Lemma~\ref{tilde oplus lemma}(d) shows
	\begin{equation*}
		(x, X) \mathbin{\tilde{\otimes}} (\infty, \emptyset) = (x + \infty, X \concat \emptyset) = (\infty, \emptyset).
	\end{equation*}
	Showing $(\infty, \emptyset) \mathbin{\tilde{\otimes}} (x, X) = (\infty, \emptyset)$ is analogous.
		
	\item Distributivity of $\tilde{\otimes}$ over $\tilde{\oplus}$: Using Lemma~\ref{tilde oplus lemma}(z, f, h) shows
	\begin{align*}
		& (x, X) \mathbin{\tilde{\otimes}} \bigl( (y, Y) \mathbin{\tilde{\oplus}} (z, Z) \bigr) \\
		& = (x, X) \mathbin{\tilde{\otimes}} \bigl( \min(y, z), \delta_{y, \min(y, z)}(Y) \cup \delta_{z, \min(y, z)}(Z) \bigr) \\
		& = \bigl( x + \min(y, z), X \concat (\delta_{y, \min(y, z)}(Y) \cup \delta_{z, \min(y, z)}(Z)) \bigr) \\
		& = \bigl( \min(x + y, x + z), \\
		& \hphantom{{} = \big(} (X \concat \delta_{y, \min(y, z)}(Y)) \cup (X \concat \delta_{z, \min(y, z)}(Z)) \bigr) \\
		& = \bigl( \min(x + y, x + z), \\
		& \hphantom{{} = \big(} \delta_{y, \min(y, z)}(X \concat Y) \cup \delta_{z, \min(y, z)}(X \concat Z) \bigr) \\
		& = \bigl( \min(x + y, x + z), \\
		& \hphantom{{} = \big(} \delta_{x + y, \min(x + y, x + z)}(X \concat Y) \\
		& \hphantom{{} = \big(} \cup \delta_{x + z, \min(x + y, x + z)}(X \concat Z) \bigr) \\
		& = (x + y, X \concat Y) \mathbin{\tilde{\oplus}} (x + z, X \concat Z) \\
		& = \bigl( (x, X) \mathbin{\tilde{\otimes}} (y, Y) \bigr) \mathbin{\tilde{\oplus}} \bigl( (x, X) \mathbin{\tilde{\otimes}} (z, Z) \bigr).
	\end{align*}
	Showing $\bigl( (y, Y) \mathbin{\tilde{\oplus}} (z, Z) \bigr) \mathbin{\tilde{\otimes}} (x, X) = \bigl( (y, Y) \mathbin{\tilde{\otimes}} (x, X) \bigr) \mathbin{\tilde{\oplus}} \bigl( (z, Z) \mathbin{\tilde{\otimes}} (x, X) \bigr)$ is analogous.
\end{itemize}

Finally, we show that $(\mathbb{S}', \oplus', \otimes', 0^{\mathbb{S}'}, 1^{\mathbb{S}'})$ is a semiring. To that effect, let $X, Y, Z \in \mathbb{S}'$ be arbitrary. We additionally extend $\otimes'$ to apply to all elements of $\powerset(\mathbb{S} \times \mathbb{S} \times \mathbb{S} \times V)$.
\begin{itemize}
	\item Associativity of $\oplus'$: Immediate.
	\item Commutativity of $\oplus'$: Immediate.
	\item Associativity of $\otimes'$: We start by proving the following lemma.
	
	\begin{lem}
		Given $X, Y \in \powerset(\mathbb{S} \times \mathbb{S} \times \mathbb{S} \times V)$, then $X \mathbin{\otimes'} (Y \cap \mathbb{T}) = X \mathbin{\otimes'} Y = (X \cap \mathbb{T}) \mathbin{\otimes'} Y$.
	\end{lem}
	\begin{proof}
		As $Y \cap \mathbb{T} \subseteq Y$, we have $X \mathbin{\otimes'} (Y \cap \mathbb{T}) \subseteq X \mathbin{\otimes'} Y$. If $(y_1, y_2, y_3, u) \in Y$ is such that either $y_1 = 0^\mathbb{S}$, $y_2 = 0^\mathbb{S}$, or $y_3 = 0^\mathbb{S}$, then $(x_1 \otimes y_1, x_2 \otimes y_2, x_3 \otimes y_3, u) \notin \mathbb{T}$ for any values of $x_1, x_2, x_3 \in \mathbb{S}$ since $0^\mathbb{S}$ is an annihilator for $\otimes$ and hence at least one of $x_1 \otimes y_1$, $x_2 \otimes y_2$, $x_3 \otimes y_3$ equal $0^\mathbb{S}$. This shows that $Y \setminus \mathbb{T}$ contributes no elements to $X \mathbin{\otimes'} Y$.
		
		Proving $(X \cap \mathbb{T}) \mathbin{\otimes'} Y = X \mathbin{\otimes'} Y$ is analogous.
	\end{proof}
	
	Using the above lemma shows
	\begin{align*}
		& X \mathbin{\otimes'} (Y \mathbin{\otimes'} Z) \\
		& = X \mathbin{\otimes'} \bigl( \{ (y_1 \otimes z_1, y_2 \otimes z_2, y_3 \otimes z_3, u) \mid \\
		& \hphantom{{} = X \mathbin{\otimes'} \bigl( \{ } (y_1, y_2, y_3, u) \in Y, (z_1, z_2, z_3, u) \in Z \} \cap \mathbb{T} \bigr) \\
		& = X \mathbin{\otimes'} \{ (y_1 \otimes z_1, y_2 \otimes z_2, y_3 \otimes z_3, u) \mid \\
		& \hphantom{{} = X \mathbin{\otimes'} \{ } (y_1, y_2, y_3, u) \in Y, (z_1, z_2, z_3, u) \in Z \} \\
		& = \{ (x_1 \otimes (y_1 \otimes z_1), x_2 \otimes (y_2 \otimes z_2), x_3 \otimes (y_3 \otimes z_3), u) \mid \\
		& \hphantom{{} = \{} (x_1, x_2, x_3, u) \in X, (y_1, y_2, y_3, u) \in Y, \\
		& \hphantom{{} = \{} (z_1, z_2, z_3, u) \in Z \} \cap \mathbb{T} \\
		& = \{ ((x_1 \otimes y_1) \otimes z_1, (x_2 \otimes y_2) \otimes z_2, (x_3 \otimes y_3) \otimes z_3, u) \mid \\
		& \hphantom{{} = \{} (x_1, x_2, x_3, u) \in X, (y_1, y_2, y_3, u) \in Y, \\
		& \hphantom{{} = \{} (z_1, z_2, z_3, u) \in Z \} \cap \mathbb{T} \\
		& = \{ (x_1 \otimes y_1, x_2 \otimes y_2, x_3 \otimes y_3, u) \mid \\
		& \hphantom{{} = \{} (x_1, x_2, x_3, u) \in X, (y_1, y_2, y_3, u) \in Y \} \mathbin{\otimes'} Z \\
		& = \bigl( \{ (x_1 \otimes y_1, x_2 \otimes y_2, x_3 \otimes y_3, u) \mid \\
		& \hphantom{{} = \bigl( \{} (x_1, x_2, x_3, u) \in X, (y_1, y_2, y_3, u) \in Y \} \cap \mathbb{T} \bigr) \mathbin{\otimes'} Z \\
		& = (X \mathbin{\otimes'} Y) \mathbin{\otimes'} Z.
	\end{align*}
	\item $0^{\mathbb{S}'} \coloneq \emptyset$ is an identity for $\oplus'$: Immediate.
	\item $1^{\mathbb{S}'} \coloneq \{ (1^\mathbb{S}, 1^\mathbb{S}, 1^\mathbb{S}, u) \mid u \in V \}$ is an identity for $\otimes'$:
	\begin{align*}
		X \mathbin{\otimes'} 1^{\mathbb{S}'} & = \{ (x_1 \otimes y_1, x_2 \otimes y_2, x_3 \otimes y_3, u) \mid \\
		& \hphantom{{} = \{} (x_1, x_2, x_3, u) \in X, (y_1, y_2, y_3, u) \in 1^{\mathbb{S}'} \} \cap \mathbb{T} \\
		& = \{ (x_1 \otimes 1^\mathbb{S}, x_2 \otimes 1^\mathbb{S}, x_3 \otimes 1^\mathbb{S}, u) \mid \\
		& \hphantom{{} = \{} (x_1, x_2, x_3, u) \in X, (1^\mathbb{S}, 1^\mathbb{S}, 1^\mathbb{S}, u) \in 1^{\mathbb{S}'} \} \cap \mathbb{T} \\
		& = \{ (x_1, x_2, x_3, u) \mid \\
		& \hphantom{{} = \{} (x_1, x_2, x_3, u) \in X, (1^\mathbb{S}, 1^\mathbb{S}, 1^\mathbb{S}, u) \in 1^{\mathbb{S}'} \} \cap \mathbb{T} \\
		& = \{ (x_1, x_2, x_3, u) \mid (x_1, x_2, x_3, u) \in X \} \cap \mathbb{T} \\
		& = X.
	\end{align*}
	Showing $1^{\mathbb{S}'} \mathbin{\otimes'} X = X$ is analogous.
	\item $0^{\mathbb{S}'} \coloneq \emptyset$ is an annihilator for $\otimes'$: 
	\begin{align*}
		X \mathbin{\otimes'} 0^{\mathbb{S}'} & = \{ (x_1 \otimes y_1, x_2 \otimes y_2, x_3 \otimes y_3, u) \mid \\
		& \hphantom{{} = \{} (x_1, x_2, x_3, u) \in X, (y_1, y_2, y_3, u) \in 0^{\mathbb{S}'} \} \\
		& = \emptyset.
	\end{align*}
	Showing $0^{\mathbb{S}'} \mathbin{\otimes'} X = 0^{\mathbb{S}'}$ is analogous.
	\item Distributivity of $\otimes'$ over $\oplus'$: Analogous to the proof of Lemma~\ref{tilde oplus lemma}(h), noting that $\cap$ distributes over $\cup$ to account for the intersections with $\mathbb{T}$.
\end{itemize}

\section*{Appendix B: Proofs of Propositions}
\label{appendix B}

\begin{proof}[Proof of Proposition~\ref{paths of optimal weight proposition}.]
	We proceed by induction on $n \geq 1$.
	\begin{description}
		\item[Base case.] When $n = 1$, we have 
		\begin{equation*}
			\mat{B}_1(u, v) = \real(\tilde{\mat{A}})(u, v) = \bigl( \mat{A}(u, v), \{\langle u, v \rangle\} \bigr).
		\end{equation*}
		The only $1$-hop path from $u$ to $v$ is $\langle u, v \rangle$, so $(\mat{A}(u, v), \{\langle u, v \rangle\})$ trivially gives the least weight of any $1$-hop path from $u$ to $v$ along with the set of all such $1$-hop paths.

		\item[Induction step.] Assume for our induction hypothesis that for all $u, v \in V$ it is the case that $\mat{B}_n(u, v) = (a, A)$ satisfies that $a$ is the least weight of any $n$-hop path from $u$ to $v$ and $A$ is the set of all such $n$-hop paths having weight $a$. Given any $u, v \in V$, 
		\begin{align*}
			\mat{B}_{n+1}(u, v) & = \bigl( \mat{B}_n \arrayprod{\tilde{\oplus}}{\tilde{\otimes}} \imaginary(\tilde{\mat{A}}) \bigr)(u, v) \\
			& = \widetilde{\bigoplus_{w \in V}}{\Bigl( \mat{B}_n(u, w) \tilde{\otimes} \bigl( \mat{A}(w, v), \{\langle v \rangle\} \bigr) \Bigr)} \\
			& = \widetilde{\bigoplus_{w \in V}}{\Bigl( (x_w, X_w) \tilde{\otimes} \bigl( \mat{A}(w, v), \{\langle v \rangle\} \bigr) \Bigr)} \\
			& = \widetilde{\bigoplus_{w \in V}}{\bigl( x_w + \mat{A}(w, v), \{\lambda \concat \langle v \rangle \mid \lambda \in X_w\} \bigr)},
		\end{align*}
		where $X_w$ is the set of minimal-weight $n$-hop paths from $u$ to $w$ and $x_w$ is that minimal weight. 
		
		For any $w \in V$ such that $(w, v) \in E$, $\{\lambda \concat \langle v \rangle \mid \lambda \in X_w\}$ is the set of all minimal-weight $(n+1)$-hop paths from $u$ to $v$ that pass through $w$ immediately before ending at $v$; moreover, the weight of all such $(n+1)$-hop paths are $x_w + \mat{A}(w, v)$. Any $(n+1)$-path $\kappa$ from $u$ to $v$ of minimal weight must pass through some $w$ immediately before ending at $v$, meaning we may write $\kappa = \lambda \concat \langle v \rangle$ for some $n$-hop path $\lambda$ from $u$ to $w$. Such a $\lambda$ must be of minimal weight amongst all $n$-hop paths from $u$ to $w$, since if $\lambda'$ were an $n$-hop path from $u$ to $w$ with strictly lower weight than that of $\lambda$, then $\lambda' \concat \langle v \rangle$ would be an $(n+1)$-hop path from $u$ to $v$ passing through $w$ immediately before ending at $v$ which has strictly lower weight than $\kappa$, contrary to hypothesis. As such, $\kappa$ has weight $\mathop{\mathrm{w}}(\lambda) + \mat{A}(w, v) = x_w + \mat{A}(w, v)$ and so $\kappa$ and its weight are represented by some summand of the summation defining $\mat{B}_{n+1}(u, v)$ above. 		
	\end{description}
\end{proof}

\begin{proof}[Proof of Proposition~\ref{catvalmul and catvalmul complex computation}.]
	Suppose $u, v \in V$. Then
	\begin{align*}
		& \bigl( \imaginary(\mat{A}') \arrayprod{\oplus'}{\otimes'} \real(\mat{B}') \bigr)(u, v) \\
		& = \bigoplus\nolimits'\{ \{(\mat{A}(u, w), 1^\mathbb{S}, \mat{A}(u, w), w)\} \\
		& \hphantom{{} = \bigoplus\nolimits'\{} \mathbin{\otimes'} \{(1^\mathbb{S}, \mat{B}(w, v), \mat{B}(w, v), w)\} \mid \\
		& \hphantom{{} = \bigoplus\nolimits'\{} w \in V, \mat{A}(u, w) \neq 0^\mathbb{S}, \mat{B}(w, v) \neq 0^\mathbb{S} \} \\
		& \hphantom{{} = } \mathbin{\oplus'} \bigoplus\nolimits'\{ \emptyset \mathbin{\otimes'} \{(1^\mathbb{S}, \mat{B}(w, v), \mat{B}(w, v), w)\} \mid \\
		& \hphantom{{} = \mathbin{\oplus'} \bigoplus\nolimits' \{} w \in V, \mat{A}(u, w) = 0^\mathbb{S}, \mat{B}(w, v) \neq 0^\mathbb{S}\} \\
		& \hphantom{{} = } \mathbin{\oplus'} \bigoplus\nolimits'\{ \{(\mat{A}(u, w), 1^\mathbb{S}, w)\} \mathbin{\otimes'} \emptyset \mid \\
		& \hphantom{{} = \mathbin{\oplus'} \bigoplus\nolimits' \{} w \in V, \mat{A}(u, w) \neq 0^\mathbb{S}, \mat{B}(w, v) = 0^\mathbb{S}\} \\
		& \hphantom{{} = } \mathbin{\oplus'} \bigoplus\nolimits'\{ \emptyset \mathbin{\otimes'} \emptyset \mid \\
		& \hphantom{{} = \mathbin{\oplus'} \bigoplus\nolimits' \{} w \in V, \mat{A}(u, w) = 0^\mathbb{S}, \mat{B}(w, v) = 0^\mathbb{S}\} \\
		& = \bigoplus\nolimits' \{ \{(\mat{A}(u, w) \otimes 1^\mathbb{S}, 1^\mathbb{S} \otimes \mat{B}(w, v), \mat{A}(u, w) \otimes \mat{B}(w, v), w)\} \\
		& \hphantom{{} = \bigoplus\nolimits' \{ } \cap \mathbb{T} \mid w\in V, \mat{A}(u, w) \neq 0^\mathbb{S}, \mat{B}(w, v) \neq 0^\mathbb{S} \} \\
		& \hphantom{{} = } \mathbin{\oplus'} \emptyset \mathbin{\oplus'} \emptyset \mathbin{\oplus'} \emptyset \\
		& = \{(\mat{A}(u, w), \mat{B}(w, v), \mat{A}(u, w) \otimes \mat{B}(w, v), w) \mid \\
		& \hphantom{{} = \{ } w \in V, \mat{A}(u, w) \otimes \mat{B}(w, v) \neq 0^\mathbb{S} \}
	\end{align*}
\end{proof}

\end{document}